\pdfoutput=1
\documentclass[12pt,draftclsnofoot, onecolumn]{IEEEtran}
\usepackage{amsthm}
\usepackage{graphicx}
\usepackage{subfigure}
\usepackage{amsmath}
\usepackage{amssymb}%
\usepackage{multicol}
\usepackage{float}
\usepackage{color}
\usepackage{array}
\usepackage{caption}
\usepackage{multirow}
\usepackage{algorithm}
\usepackage{algorithmic}
\usepackage{epstopdf}
\usepackage{cases}
\usepackage{url}
\usepackage[numbers]{natbib}

\makeatletter

\newcommand{\Rmnum}[1]{\expandafter\@slowromancap\romannumeral #1@}
\makeatother

\newtheorem{theorem}{Theorem}
\newtheorem{lemma}{Lemma}
\newtheorem{proposition}{Proposition}

\newtheorem{corollary}{Corollary}

\usepackage[T1]{fontenc}
\usepackage{ifpdf}
 \ifpdf
 \else
 \fi

\begin{document}
%
\title{Fast Deployment of UAV Networks for Optimal Wireless Coverage}
\author{\IEEEauthorblockN{Xiao~Zhang~and~Lingjie Duan}
\thanks{

X. Zhang and L. Duan are with Engineering Systems and Design Pillar, Singapore University of Technology and Design, Singapore.
E-mail: \{zhang\_xiao, lingjie\_duan\}@sutd.edu.sg.
} }

\maketitle


\begin{abstract}
Unmanned Aerial Vehicle (UAV) networks have emerged as a promising technique to rapidly provide wireless coverage to a geographical area, where a flying UAV can be fast deployed to serve as cell site. Existing work on UAV-enabled wireless networks overlook the fast UAV deployment for wireless coverage, and such deployment problems have only been studied recently in sensor networks. Unlike sensors, UAVs should be deployed to the air and they are generally different in flying speed, operating altitude and wireless coverage radius. By considering such UAV heterogeneity to cover the whole target area, this paper studies two fast UAV deployment problems: one is to minimize the maximum deployment delay among all UAVs (min-max) for fairness consideration, and the other is to minimize the total deployment delay (min-sum) for efficiency consideration. We prove both min-max and min-sum problems are NP-complete in general. When dispatching UAVs from the same location, we present an optimal algorithm of low computational complexity $O(n^2)$ for the min-max problem. When UAVs are dispatched from different locations, we propose to preserve their location order during deployment and successfully design a fully polynomial time approximation scheme (FPTAS) of computation complexity $O(n^2 \log \frac{1}{\epsilon})$ to arbitrarily approach the global optimum with relative error $\epsilon$. The min-sum problem is more challenging. When UAVs are dispatched from the same initial location, we present an approximation algorithm of linear time. As for the general case, we further reformulate it as a dynamic program and propose a pseudo polynomial-time algorithm to solve it optimally.
\end{abstract}


\IEEEpeerreviewmaketitle

%

\vspace{3cm}

\IEEEraisesectionheading{\section{Introduction}\label{sec:intr}}
Recent years have witnessed increasingly more exercises and uses of Unmanned Aerial Vehicle (UAV) networks for rapidly providing wireless coverage \cite{zeng2016wireless}. This flying cell site technology enabled by UAV rapidly expands the wireless coverage of the static territorial base stations on the ground, where UAVs serve as flying base stations to serve a geographical area (e.g., a disaster zone) out of the reach of the cellular networks. For example, Verizon has developed airborne LTE service allowing communications between a UAV and hurricane disaster victims \cite{ref:Verizon}. Moreover, Project Loon \cite{ref:loon} uses balloons as flying base stations to provide high speed internet coverage to people in rural and remote areas worldwide. In addition, traditional base stations or access points \cite{liu2012traffic} are deployed at fixed locations for a long term by meeting the average traffic load, while flying UAVs are mobile and do not have such constraint to meet varying traffic load \cite{zhang2017spectrum}. Thanks to such advantage, wireless carriers such as AT\&T have started to use UAVs to opportunistically boost wireless coverage for crowds in big concerts or sports, where people continuously post their selfies and videos online \cite{ref:att}.

There is increasingly more research work to deploy UAVs for providing wireless coverage (e.g., \cite{zeng2016wireless, al2014optimal, mozaffari2015drone, zeng2017energy}). For example, \cite{al2014optimal} and \cite{mozaffari2015drone} {consider the scenario that the wireless communication channels between UAVs and ground users are dominated by both line-of-sight (LoS) and Non-line-of-sight (NLOS) links.} They investigate the optimal operating altitude for a single UAV, where a larger UAV height increases the line-of-sight opportunity of air-to-ground transmission but incurs a larger path loss. In a UAV-enabled wireless network, \cite{zeng2017energy} {adopts the LoS dominated communication and} studies the tradeoff between a UAV’s energy consumption and communication throughput by optimizing the UAV’s moving trajectory. However, the existing work on UAV networks overlook the fast UAV deployment problems to rapidly provide the wireless coverage. Only some recent work about sensor networks study the deployment problems (e.g., \cite{wang2015minimizing}). Such results cannot apply to our fast UAV deployment problems. Unlike sensors or traditionally vehicular networks \cite{pang2014efficient} \cite{pan2012cooperative}, UAVs should be deployed to the air and the optimal deployment should take into account their heterogeneity in flying speed, operating altitude and wireless coverage radius.


Given the aforementioned limitations, we advance the research on fast deployment of heterogeneous UAVs. In practice, UAVs have limited coverage radii and can only serve users closely. Before servicing its associated users, each UAV takes the travel time or deployment delay to reach its final position and the delay depends on the travel distance to its final operational position, flying speed and operating altitude. As reported in \cite{bedfordunmanned}, different types of UAV have different mission altitudes, radii, flying speeds and endurance. For example, \emph{Micro} UAV's altitude is generally smaller than 400 feet, flying speed is from 10 to 25 miles/hour, radius is from 1 to 5 miles and endurance is about 1 hour. {By considering such UAV heterogeneity and focusing on the LoS dominated communication scenario to cover the whole target area}, we comprehensively study two fast deployment problems: one is to minimize the maximum deployment delay among all UAVs for fairness consideration and the other is to minimize the total deployment delay for efficiency consideration. The min-max optimization problem arises naturally in situations of disasters or battle fields when we fairly care about the service delivery delay to any potential wireless user in the target region. A disaster victim or soldier may appear in any location of the target region and the min-max problem targets at minimizing the worse-case delay performance of any user. We want to avoid the unfair deployment outcome that some users are served shortly while some others start services after a long time.

Different from the min-max optimization problem, the min-sum problem targets at minimizing the sum of all UAVs' travel time, or equivalently, the average travel time to attain a full coverage of the target region. This efficiency problem arises naturally in a situation when we aim to service many users in a big concert or sport and care the average waiting time performance of the crowd rather than the performance difference between individuals. Minimizing the total delay helps to improve the average service quality. Both the min-max and min-sum problems are important for different scenarios and they are conflicting in nature. On one hand, minimizing the maximum deployment delay may imply a significant increase in the total deployment delay. On the other hand, minimizing the total deployment delay may imply a significant increase in maximum deployment delay, since it does not consider reducing deployment delays of all UAVs in a fair manner. Designing different algorithms for both problems is well-motivated.



{\em Our key novelty and main contributions are summarized as follows.}
\begin{itemize}
\item \emph{Novel UAV fast deployment for Wireless Coverage (Section~\ref{sec:sys}):} To our best knowledge, this is the first paper to study heterogenous UAV deployment for providing emergent wireless coverage to a target geographical area. We prove that the both problems with objectives of min-max and min-sum deployment delay are NP-complete in general.
{\item \emph{Minimizing maximum UAV deployment delay (Section~\ref{sec:max}):}
    \begin{itemize}
    \item When a number $n$ of diverse UAVs are dispatched from the same initial location (e.g., the closest UAV station) to the target area, we present an optimal deployment algorithm of low computational complexity $O(n^2)$ by balancing UAVs' diverse flying speeds and coverage radii.
    \item When UAVs are generally dispatched from different locations, we propose to preserve their location order during deployment and successfully design a fully polynomial time approximation scheme (FPTAS) of computation complexity $O(n^2 \log \frac{1}{\epsilon})$ to arbitrarily approach the global optimum with relative error $\epsilon$.
    \end{itemize}
\item \emph{Minimizing total UAV deployment delay (Section~\ref{sec:sum}):}
    \begin{itemize}
    \item When UAVs are dispatched from the same initial location, we present a linear time approximation algorithm with provable performance bound.
    \item When dispatching UAVs from different locations, we further reformulate the min-sum problem as a dynamic program and propose a pseudo polynomial-time algorithm to solve it optimally.
    \end{itemize}}
\end{itemize}
\section{Related Work}

The use of UAVs as flying base stations is attracting growing interests from researchers \cite{mozaffari2017mobile} \cite{zhang2017optimization} \cite{xuuav}. The literature on UAV-enabled communications focus on developing the air-to-ground transmission model and explore the line of sight opportunity \cite{al2014optimal} \cite{mozaffari2015drone}. {Further, Azari et al. \cite{azari2017coverage} consider the co-channel interference effect and study the UAV coverage maximization problem.}

With respect to the UAV network deployment, most of existing works investigate the deployment or movement schemes of UAVs for reducing transmit power consumption \cite{li2016energy} \cite{wu2017joint} or the propulsion energy consumption \cite{zeng2017energy} \cite{di2015energy}. Specifically, in \cite{li2016energy}, Li et al. present a UAV energy-efficient relaying system to guarantee the success rate such that the lifetime is maximized. In this system, the transmission schedule of the UAVs is optimized to reduce the maximum energy consumption of the UAVs, thereby extending its lifetime. In \cite{wu2017joint}, Wu et al. use UAVs as flying base stations to serve a group of users fairly for transmission throughput. They optimize the multiuser communication scheduling jointly with the UAVs’ trajectory and power control. In \cite{zeng2017energy}, Zeng and Zhang present a UAV propulsion energy consumption model and optimize the UAVs' coverage radii and flying speeds to maximize the energy efficiency for communication. By deriving an energy consumption model from real measurements \cite{di2015energy}, Carmelo and Giorgio optimize the UAV path to minimize the energy consumption such that all points of a specific area is covered. In \cite{orfanus2016self}, Orfanus et al. use multiple UAVs as relay nodes in the self-organizing paradigm to support military operations. \cite{alzenad20173d} uses the UAV as base station to provide wireless serivce to low-mobility ground users with QoS requirements, in which they aim to maximize the number of covered users. \cite{mozaffari2016efficient} derives the wireless coverage probability for UAVs as a function of the operating altitude and the antenna gain. Then, it presents a deployment scheme to maximize the coverage performance with minimum transmit power. Few of the existing work study the fast UAV deployment for providing wireless coverage.

Related to the fast UAV deployment, there are only a few recent theoretical works on sensor networks (e.g., \cite{wang2015minimizing} \cite{benkoczi2016}). These work focus on minimizing the sensors' maximum or total moving/deployment distance in the one-dimensional ground. Wang and Zhang \cite{wang2015minimizing} assume an identical sensing range for all sensors and present the first exact algorithm to compute the maximum weighted movement of sensors, which has the computation complexity $O(n^2 \log n \log \log n)$. To deal with the more general case of diverse sensing ranges and even weights for sensors, Benkoczi et al. \cite{benkoczi2016} strongly assume all sensors are on one end of the target interval and thus present an approximation algorithm to minimize the total weighted movement. However, the above algorithm design methods about sensor networks cannot apply to our fast UAV deployment problems, where UAVs should be deployed to the air and the optimal deployment should take into account their heterogeneity in flying speed, operating altitude and wireless coverage radius. Regarding to the short delay wireless service by UAVs, Mohammad et al. \cite{mozaffari2016unmanned} consider a system of UAV with underlaid Device-to-Device communications and study the tradeoff between the coverage and delay. None of the existing work study the fast UAV deployment for providing full wireless coverage over a target area.

\section{System Model and Problem Formulation}\label{sec:sys}

\begin{figure}[t]
    \centering
        \includegraphics[width=0.75\textwidth]{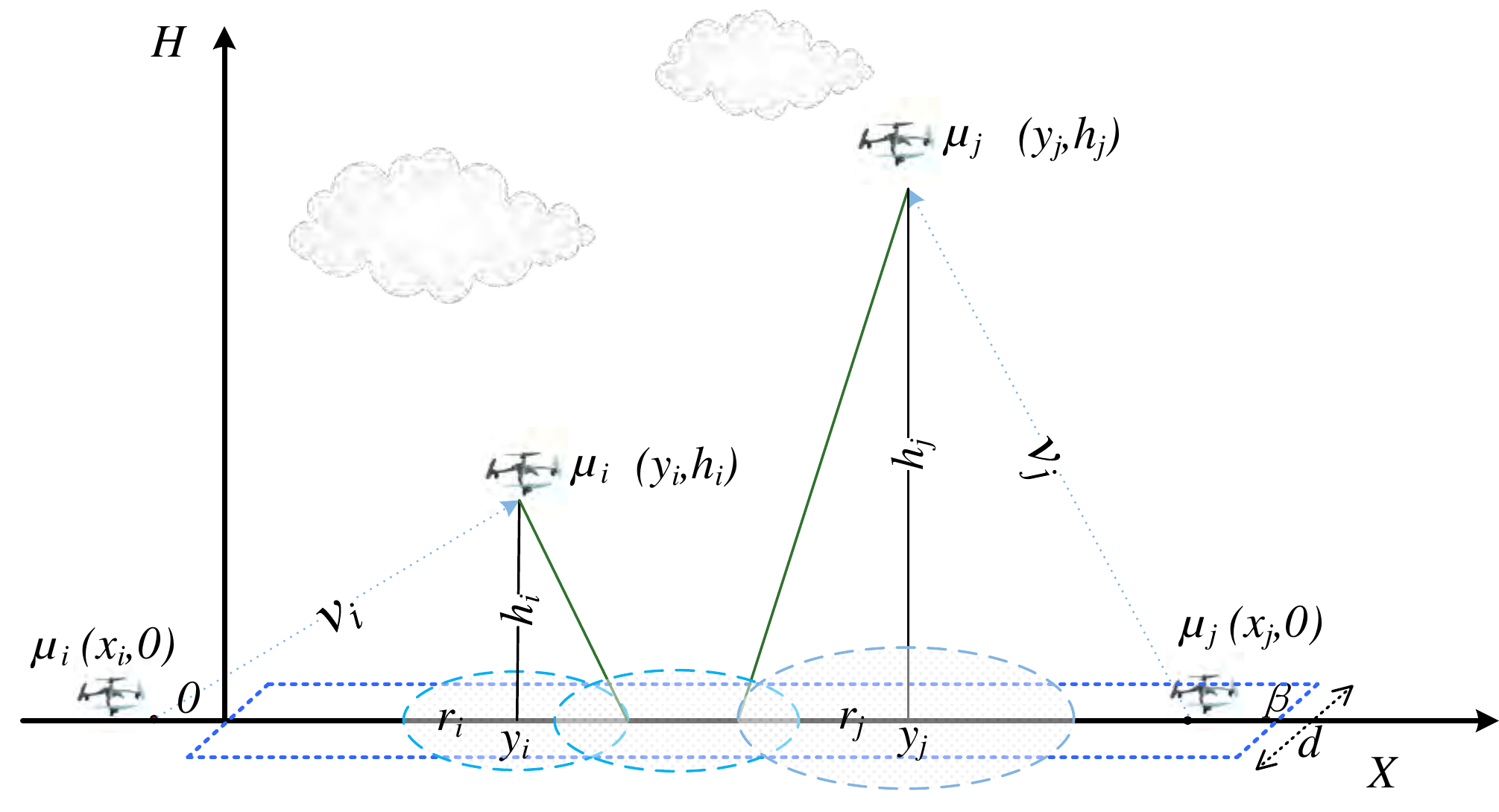}
    \caption{System model for deploying UAVs to provide wireless coverage to the target area \textbf{A} ($\beta, d$), where UAV $\mu_i$ with coverage radius $r_i$ is deployed from $x_i$ initially to $y_i \in [0, \beta]$ at operating altitude $h_i$.}
    \label{fig:covermodel}
\end{figure}
%

This section introduces our system model and problem formulation for deploying UAVs to provide wireless coverage to the whole target.

As shown in Figure~\ref{fig:covermodel}, a centralized system needs to emergently deploy UAVs to provide full wireless coverage over the target area \textbf{A}, which is a rectangle with length $\beta$ and width $d$ as in Equation (\ref{eq:A}). Restricting the target area to be a thin rectangular area is reasonable to mimic avenues, streets and highways. The notations and corresponding meanings are given in Table~\ref{tab:notation}. The UAVs in a set $\textbf{U} = \{\mu_1, \cdots, \mu_n\}$ are initially located in different locations $\{x_1, \cdots, x_n\}$ along $x$-axis (e.g., in ground UAV stations) before the deployment. Without loss of generality, we assume $x_1 \leq \cdots \leq x_n$. We denote a UAV $\mu_i$'s final position after deployment as $(y_i, h_i)$ at operating altitude $h_i$. The UAV $\mu_i$ flies from its initial location $(x_i, 0)$ to its designed destination $(y_i, h_i)$ and then hovers at the operating altitude $h_i$ to serve the ground users.

{
\begin{equation}
\label{eq:A}
\textbf{A} = \{ (w,l)| 0 \leq w \leq d, 0 \leq l \leq \beta \}.
\end{equation}
}

\begin{table}[t]%
\centering
\caption{Notations and their physical meanings.}
\label{tab:notation}%
\begin{tabular}{|l|l|}
\hline
Notation                   & Meaning\\\hline
$n$                         & Number of UAVs\\\hline
$\mu_{i}$                   & Index of UAV $i$\\\hline
$\beta$                     & Rightmost endpoint of the interval $L =[0, \beta]$\\\hline
$x_{i}$                     & Initial location of UAV $i$ before deployment \\\hline
$y_{i}$                     & Final location of UAV $i$ after deployment \\\hline
$r_{i}$                     & Coverage radius of UAV $i$\\\hline
$h_{i}$                     & Operating altitude of UAV $i$\\\hline
$v_{i}$                     & Flying speed of UAV $i$\\\hline
$T$                         & Maximum deployment delay obtained by Algorithm~\ref{alg:fptas}     \\\hline
$T^*$                       & Minimum maximum deployment delay of problem (\ref{IOP})                     \\\hline
$\Gamma^{\prime}$           & Total deployment delay obtained by Algorithm~\ref{alg:fptas}       \\\hline
$\Gamma^*$                  & Minimum total deployment delay $\Gamma^{*} = \sum_{1 \leq i \leq n} T_i^*$  \\\hline
$\epsilon$                  & The relative error in Algorithm~\ref{alg:fptas}                             \\\hline
\end{tabular}
\end{table}

We assume that the UAVs have sufficient bandwidth resources so that all UAVs can be assigned orthogonal channels and for avoiding interference-free. This interference-free model is widely used (e.g. \cite{mozaffari2017mobile} \cite{zhan2011wireless}). In practice the assigned channels for distant UAVs can be reused during deployment. Thus, we assume the interference among UAVs can be ignored, and henceforth we focus our study on dealing with the UAV coverage issue.

As in \cite{zeng2017energy}, we adopt the air-to-ground model where the wireless communication channels between UAVs and ground users in the target area are dominated by LoS links. LoS links are expected for air-to-ground channels in many scenarios [1]. Therefore, the channel power gain from the UAV to each user $k$ is modeled as the free-space path loss model, i.e., $g_k = \xi \bar{d_k}^{-2}$, where $\xi$ denotes the channel power gain at a reference distance. $\bar{d_k}$ is the link distance between the UAV and ground user $k$. Given a standard transmission power $P_i$, the signal-to-noise ration (SNR) at ground user $k$ is given by $\gamma_k = \frac{P_i g_k}{\sigma^2}$, where $\sigma^2$ denotes the noise power at each ground user. We say a ground user $k$ is covered by a UAV if the SNR at user $k$ is no less than a threshold value $\gamma_{th}$, which indicates the target data rate of each user is at least $1+\log (1+ \gamma_{th})$. Thus, we can obtain the correlation between a UAV's wireless coverage range $r$ (i.e., the maximum ground range for just achieving the threshold SNR $\gamma_{th}$) and operating altitude $h$ as $\frac{P_i \xi}{\sigma^2 (r^2 + h^2)} = \gamma_{th}$ \footnote{Such correlations as in \cite{zeng2017energy} \cite{azari2017ultra} will not affect our deployment algorithms design later, in which we consider a more general case that UAVs have arbitrary parameters of operating altitude, flying speed and coverage radius.}. We have $r_i = \sqrt{\frac{P_i \xi}{\gamma_{th} \sigma^2} - h_i^2}$ for UAV $\mu_i$.

A particular UAV $\mu_i$ operating at final position $(y_i, h_i)$ covers a region $D_i$ in the target rectangular, which is defined in Equation (\ref{eq:area}).

{
\setlength{\abovedisplayskip}{3pt}
\setlength{\belowdisplayskip}{-1pt}
\begin{equation}
\label{eq:area}
\begin{split}
D_i =& \{(w,l)|-\frac{d}{2} \leq w \leq \frac{d}{2}, \\
&y_i-\sqrt{r_i^2 - (\frac{d}{2})^2} \leq l \leq y_i+\sqrt{r_i^2 - (\frac{d}{2})^2}\}.
\end{split}
\end{equation}}

We require a full coverage over the target target area \textbf{A} ($\beta, d$) by deploying $n$ diverse UAVs, i.e., $\textbf{A} (\beta, d) \subseteq  \bigcup_1^n D_i$.




During the deployment, UAV $\mu_i$ travels an Euclidean distance $\sqrt{(y_i-x_i)^2+h_i^2}$ at constant flying speed $v_i$ as in \cite{hencheyflight}. Thus, its travel time is given by\footnote{This model is flexible to keep generality. Besides using the Euclidean distance, it can also be extended to (weighted) Manhattan distance $|x_i-y_i|+h_i$.}
{
\begin{equation}
\label{eq:time}
{T}_i(y_i) = \frac{\sqrt{(y_i - x_i)^2 + h_i^2}}{v_i}.
\end{equation}}

After considering all UAVs' travel time, we define the maximum deployment delay as the maximum travel time among all UAVs till reaching the full coverage over the target area \textbf{A}. Our maximum deployment delay optimization problem is thus
{
\begin{align}
\label{general}
&~ \min_{\{y_1,\cdots, y_n\}} \underset{1 \leq i \leq n} {\max} ~ {T}_i(y_i) ~ , \\
&~ {\rm s.t.}, \ \textbf{A} \subseteq  \bigcup_1^n D_i.\notag
\end{align}}

\noindent Note that the min-max objective is to balance the deployment time among all UAVs and fairly optimize the delay bottleneck for reaching the full coverage of the target.

In addition, we further consider the total deployment delay objective as the summation of travel times of all UAVs till reaching the full coverage over the target interval\footnote{{Although the energy efficiency is not considered explicitly here, our deployment delay objective does not conflict with energy efficiency concern. For example, in our min-sum delay optimization problem, we reduce the total travel time for all UAVs such that the target is fully covered. After this deployment, more UAVs' residual energy for hovering and serving ground users will be saved.}}. Under the efficiency consideration, our total deployment delay optimization problem is thus
{
\begin{align}
\label{sum}
&~ \min_{\{y_1,\cdots, y_n\}} \underset{1 \leq i \leq n} {\sum} ~ {T}_i(y_i) ~ , \\
&~ {\rm s.t.}, \ \textbf{A} \subseteq  \bigcup_1^n D_i.\notag
\end{align}}

We assume $2 \sum_{i=1}^{n}{\sqrt{r_i^2 - (\frac{d}{2})^2}} \geq \beta$ and $r_i \geq \frac{d}{2}$ throughout the paper. Otherwise, there is no feasible deployment to problems (\ref{general}) and (\ref{sum}). Since the width $d$ of the rectangular area \textbf{A} is a given constant value, we can restrict $d \to 0$ to facilitate the theoretical analysis in the following. Thus, the target region becomes a line interval $L = [0, \beta]$ as in \cite{Fan2014a} \cite{lyu2016cyclical}. We show later in Sections~\ref{sec:max} and~\ref{sec:sum} that our problems on this line interval are already NP-complete, which can shed light on the solutions to a rectangular target area \textbf{A}. Note that the minimum total deployment delay in problem (\ref{general}) is not smaller than that in problem (\ref{sum}), while the maximum deployment delay in problem (\ref{sum}) is not smaller than that in problem (\ref{general}).

Based on above problem formulations, the two fast UAV deployment problems belong to the domain of combinatorial optimization. The optimal UAV deployment is a specific combination of ordered UAVs, which is generally exponential in the number of UAVs. For a combinatorial optimization problem, theoretical insights of tractability and algorithmic results are the main concerns for problem solution. Though we consider the simplest possible line interval for target area as in \cite{lyu2016cyclical}, the fast UAV deployment problems (by considering different coverage radii $r_i$'s, operating altitudes $h_i$'s and flying speeds $v_i$'s are beyond prior deployment literature's methods for homogeneous sensor networks (i.e., \cite{lee2017minimizing} \cite{Fan2014a}). As we will show later in Sections~\ref{sec:max} and~\ref{sec:sum}, both problems (\ref{general}) and (\ref{sum}) are actually NP-complete. They are difficult to solve due to the UAVs' distinct initial locations and their multi-dimensional heterogeneity, which result in exponential number of sequences and combinations of UAVs.


%
%

\section{Optimization of min-max deployment problem} \label{sec:max}

In this section, we investigate how to dispatch a number $n$ of UAVs in a fair manner by targeting at the deployment delay to a user location in the worst case. Our min-max problem in~(\ref{general}) aims to minimize the maximum deployment delay among all UAVs such that any possible user located in the target region \textbf{A} ($\beta, d$) is treated fairly.


In the following, we first show that the problem (\ref{general}) when UAVs are dispatched from different locations (i.e., $x_i \neq x_j$) is NP-complete by reduction from the classic {\em 3-partition problem} \cite{garey2002computers}.

\begin{theorem} \label{thry:np1}
The min-max deployment delay problem in (\ref{general}) is NP-complete.
\end{theorem}
\begin{proof}
{See Appendix \ref{ap:npc1}.}
\end{proof}

\subsection{Optimal UAV deployment from the same location} \label{sec:same}

\begin{algorithm}[t]
\caption{Optimal UAV dispatching algorithm from the same location}
\begin{algorithmic}[1]

\STATE \textbf{Input:}\\
$\textbf{U} =\{\mu_1,\mu_2, \ldots, \mu_n\}$ \\

\STATE \textbf{Output:}\\
$y_{i}^*$: final location of $\mu_{i}$ \\
${T}$: optimal deployment delay

\STATE $\overline{\beta} = \beta, \textbf{U}^- = \textbf{U}, \textbf{T} = \emptyset$
\WHILE{$\overline{\beta} > 0$}
\STATE $\mu_j \gets \arg \min_{\mu_i \in \textbf{U}^-} \frac{\sqrt{(\overline{\beta}-r_i)^2 + h_i^2}}{v_i}$
\STATE {$\textbf{T} \gets \textbf{T} \cup T_{j} = \frac{\sqrt{(\overline{\beta} - r_j)^2 + h_j^2}}{v_j}$}
\STATE {$\overline{\beta} \gets \overline{\beta} - 2r_j$}
\STATE {$y_{j} \gets \overline{\beta} - r_j$ , \  $\textbf{U}^- \gets \textbf{U}^- \setminus \{\mu_j\}$}
\ENDWHILE
\RETURN {${T} \gets \max \textbf{T}$}
\end{algorithmic}
\label{alg:common}
\end{algorithm}

We first study a special case of problem (\ref{general}) by dispatching the UAVs from the same initial location (i.e., $x_i = x_j$ for $1 \leq i, j \leq n$). Without loss of generality, we assume that $x_i \leq 0$, $\forall \ \mu_i,  1 \leq i \leq n$, which is symmetric to the case of $x_i \geq \beta$. Note that for the case of $0 < x_i < \beta$, we can divide the line interval into two subintervals, i.e., $[0, x_i]$ and $[x_i, \beta]$, and apply our deployment algorithm (as presented later) similarly over both subintervals. 

If a UAV has the larger the distance from the initial location to the target position, it incurs a larger travel time. Among all UAVs, we first consider which UAV to send and cover the furthest point of the target area. Specifically, given the current uncovered line interval ($[0, \beta]$ initially or uncovered subinterval), we sequentially select an unassigned UAV (e.g., $\mu_i$) with the minimum travel time to just cover the furthest point on the remaining uncovered interval during deployment. In our problem, though we dispatch all UAVs simultaneously, it is equivalent to dispatching of UAVs one by one to cover the line interval $[0, \beta]$. We only count and compare each UAV's travel time to calculate the maximum delay objective.

As shown in Algorithm~\ref{alg:common}, initially, we set $\overline{\beta}= \beta$ and the available (unassigned) UAV set $\textbf{U}^-=\textbf{U}$ , as we haven't sent any UAV to cover any point in the line interval yet. In each iteration, we dispatch a UAV $\mu_j$ with the minimum travel distance $T_j=\frac{\sqrt{(\overline{\beta}-r_j)^2 + h_j^2}}{v_j}$ in the available UAV set $\textbf{U}^-$ to new position $(\overline{\beta}-r_j, h_j)$. Then the uncovered interval decreases from $[0,\overline{\beta}]$ to $[0, \overline{\beta}-2r_j]$. We record $\mu_j$'s travel time $T_j$ into set $\textbf{T}$ and remove UAV $\mu_j$ from $\textbf{U}^-$. We continue to dispatch another UAV until the target interval is fully covered. In the end, we obtain the maximum deployment delay $\max{\textbf{T}}$ as the optimum $T$. Note that given the UAVs operating altitudes, we prefer to deploy those UAVs with larger flying speeds and larger coverage radius further away in the target area.


\begin{proposition}\label{common}
Algorithm~\ref{alg:common} optimally solves the min-max deployment problem in (\ref{general}) when dispatching $n$ UAVs from the same location. Its computation complexity is $O(n^2)$.
\end{proposition}
\begin{proof}
See Appendix \ref{ap:common}.
\end{proof}
%
%
%
\subsection{Problem reformulation under order preserving of UAVs' locations} \label{sec:iop}

Since problem (\ref{general}) is NP-complete generally, there is no efficient algorithm to find the optimal solution, unless P = NP. Accordingly, we propose that the UAVs preserve their initial locations' order during the deployment. Without loss of generality, we assume $x_1 \leq x_2 \ldots \leq x_n$. Given the order $\propto$ according to the UAVs' initial locations $x_i$'s, the final locations $y_i$'s of UAVs must meet the requirement: $y_i \leq y_j$ if and only if $x_i \leq x_j$. This simplifies the coordination among UAVs, and thus will simplify the algorithm design later. In practice, this is reasonable as it avoids any possible collision when two UAVs cross each other to reach their final positions \cite{Mahjri2017SLIDE}.

Our optimization problem is to minimize the deployment delay for reaching full coverage of the target area subject to the order $\propto$, i.e.,
{
\begin{align}
\label{IOP}
&~ \min_{\{y_1,\cdots, y_n\}} \underset{1 \leq i \leq n} {\max} ~ {T}_i(y_i) ~ , \\
&~ {\rm s.t.}, \ [0,\beta] \subseteq  \bigcup_1^n [y_i - r_i, y_i + r_i],\notag \\
&~~~~~~~ y_i \leq  y_{i+1},~ \forall \ 1 \leq i \leq n-1. \notag
\end{align}}
\noindent Note that the last inequality is due to location order preserving and $x_1 \leq \cdots \leq x_n$. This simplified problem is still difficult to solve since selecting a specific combination of UAVs as the optimal UAV deployment is generally exponential in the number of UAVs. In the following, we first introduce the feasibility checking problem for problem (\ref{IOP}) and design the corresponding algorithm to determine whether we can find a deployment scheme within the deadline. Then, we use binary search over those feasible deadlines to find the minimum deployment delay (deadline).

\subsubsection{Feasibility checking problem}

\begin{algorithm}[t]
\caption{Feasibility checking algorithm}
\begin{algorithmic}[1]

\STATE \textbf{Input:}\\ 
$\textbf{U} =\{\mu_1,\mu_2, \ldots, \mu_n\}$ \\
$T$: a given deployment delay deadline for all UAVs

\STATE \textbf{Output:}\\ 
$y_{i}$: final locations of $\mu_{i}$

\STATE Compute $a_{i}$ in equation (\ref{eq:a}) and $b_{i}$ in equation (\ref{eq:b}) if $v_i T \geq h_i$
\STATE $\overline{\beta} = 0$; $\textbf{U}^- = \textbf{U}$; $\textbf{S}^{c} = \emptyset$

\FOR { $i = 1$ to $n$ }
\IF {$\overline{\beta} \notin [a_i, b_i]$ or $v_i T < h_i$}
	\STATE {$\textbf{U}^- \gets \textbf{U}^{-} \setminus \{\mu_{i}\}$}
\ELSE
	\STATE {$y_{i} \gets \min\{\overline{\beta} + r_{i}, b_{i} - r_i\}$}
	\STATE {$\textbf{S}^{c} \gets \textbf{S}^{c} \cup \{u_{j} \in U^{-}: j < i, y_{j} > y_{i}\}$}
	\STATE {$\textbf{U}^- \gets \textbf{U}^{-} \setminus \textbf{S}^{c}$, \ $\overline{\beta} \gets y_{i}+r_i$}
\ENDIF

\IF {$\overline{\beta} < \beta$}
	\STATE {Break;}\\
\ENDIF
\ENDFOR

\IF {$\overline{\beta} < \beta$}
	\RETURN $T \ is \ not feasible$ ($T < T^*$)\\
\ELSE
	\RETURN $T \ is \ feasible$ ($T \geq T^*$)\\
\ENDIF

\end{algorithmic}
\label{alg:decalg}
\end{algorithm}

We first define the feasibility checking problem as follows: given any deployment delay $T>0$ and order requirement $\propto$, determine whether UAVs can be moved to reach a full coverage within deadline $T$. Let $T^*$ denotes the optimal deployment delay of problem (\ref{IOP}), we next design a feasibility checking algorithm to determine whether ${T} \geq {T}^*$ or whether such $T$ is feasible to achieve via UAV dispatching.

Consider any ${T} > 0$, for UAV $\mu_i \in \textbf{U}$ with altitude $h_i$, if $v_i \cdot T \geq h_i$, $\sqrt{(v_i T)^2 - h_i^2}$ is the maximum \emph{horizontal} distance to move on $L = [0, \beta]$. We define $a_i$ as the leftmost point and $b_i$ as the rightmost point on $L$ that can be covered by $\mu_i$ within $T$. We call $a_i$ (resp., $b_i$) the {\em leftmost (resp., rightmost) $T$-coverable point} of $\mu_i$. Then we have

{
\begin{align}
\label{eq:a}
a_i=x_i-r_i-\sqrt{(v_i T)^2 - h_i^2},
\end{align}
}
{
\begin{align}
\label{eq:b}
b_i=x_i+r_i+\sqrt{(v_i T)^2 - h_i^2}.
\end{align}
}
Algorithm~\ref{alg:decalg} solves the feasibility checking problem. It first computes $a_i$ and $b_i$ in equations (\ref{eq:a}) and (\ref{eq:b}), then deploys the UAVs one by one according to the order $\propto$ from the left endpoint of target interval $[0, \beta]$. As $x_1 \leq x_2 \leq \cdots \leq x_n$, we start with UAV $\mu_1$ and end up with $\mu_n$. Given our current covered interval $[0, \overline{\beta}]$ where the boundary $\overline{\beta} < \beta$, iteration $i$ starts with checking whether UAV $\mu_i$ can fly to altitude $h_i$ (i.e., $v_i \cdot T \geq h_i$) or not.
\begin{itemize}
\item If $v_i T < h_i$, we will not consider dispatching UAV $\mu_i$.
\item If $v_i T \geq h_i$, we still need to check if $\mu_i$ can seamlessly cover the point $\overline{\beta}$ (i.e., $\overline{\beta} \in [a_i, b_i]$). If this also holds, we will efficiently deploy $\mu_i$ to $y_i = \min(\overline{\beta}+r_i, b_i- r_i)$.
\end{itemize}

Noted that once $\mu_i$ is deployed to the left of UAV $\mu_j$, in which $j < i$, then Algorithm~\ref{alg:decalg} in line 10 will undo dispatching of $\mu_j$ and will not use this UAV. After a successful dispatching of UAV $\mu_i$, the covered interval prolongs from $[0, \overline{\beta}]$ to $[0,y_i+r_i]$ in this iteration.

If $T$ is feasible ($T \geq T^*$), our algorithm will return a subset $\textbf{U}^-$ of UAVs and their new locations $y_i$'s to fully cover target $L$ within $T$. For each UAV $\mu_i \in \textbf{U} \setminus \textbf{U}^-$, it will not be used and just stay at the initial location.

\begin{proposition} \label{lm:decalg}
The feasibility checking problem for a particular deadline is optimally solved by Algorithm~\ref{alg:decalg} in $O(n^2)$ time.
\end{proposition}
\begin{proof}
See Appendix~\ref{ap:dec}.
\end{proof}

Remark that the feasibility checking problem has independent interest because it characterizes the minimization problem model, in which each UAV has the same deployment delay deadline $T$ and we want to know whether they can move to reach a full coverage.

\subsubsection{Binary search over feasible deadlines}
\label{sec:bs}

With the help of Algorithm~\ref{alg:decalg}, we can verify whether a given deadline $T$ is feasible or not. The minimum deadline among all feasible ones is actually the optimum of problem (\ref{IOP}). Here, we apply binary search to find the minimum deadline and solve problem (\ref{IOP}). Before the search, we still need to determine the search scope and step of $T$.

For each single UAV $\mu_i$, the minimum moving distance is altitude $h_i$. Thus, the lower bound of $T$ (denoted as $T_l$) among all UAVs can be determined according to
{
\begin{align}
\label{eq:low}
T_l = \min_{1 \leq i \leq n} \frac{h_i}{v_i}.
\end{align}
}

In general, $T_l$ is not feasible because it is the minimum possible travel time among all UAVs. We next determine the upper bound of $T$ (denoted as $T_u$). For UAV $\mu_i$, the maximum possible moving distance of $\mu_i$ is to reach position $(0, h_i)$ or $(\beta, h_i)$ beyond the leftmost or rightmost location on the target interval $L = [0,\beta]$. Thus, the upper bound of $T$ (denoted as $T_u$) among all UAVs is given by

{
\begin{align}
\label{eq:upp}
T_u = \max_{1 \leq i \leq n} \{\frac{\max \{\sqrt{(\beta - x_i)^2 + h_i^2}, \sqrt{x_i^2 + h_i^2}\}}{v_i}\}.
\end{align}
}

In the binary search, we define the relative error as $\epsilon$ which is a small constant value, and accordingly set the search accuracy as $\epsilon T_l$. As illustrated in Figure~\ref{fig:binary_search}, the binary search starting with $T_l$ stops once switching from infeasible deadline $T^{\prime}$ to feasible $T^{\prime\prime}$, such that the resultant $T^{\prime\prime}$ is our searched optimum for problem (\ref{IOP}).

\begin{figure}[t]
    \centering
        \includegraphics[width=0.75\textwidth]{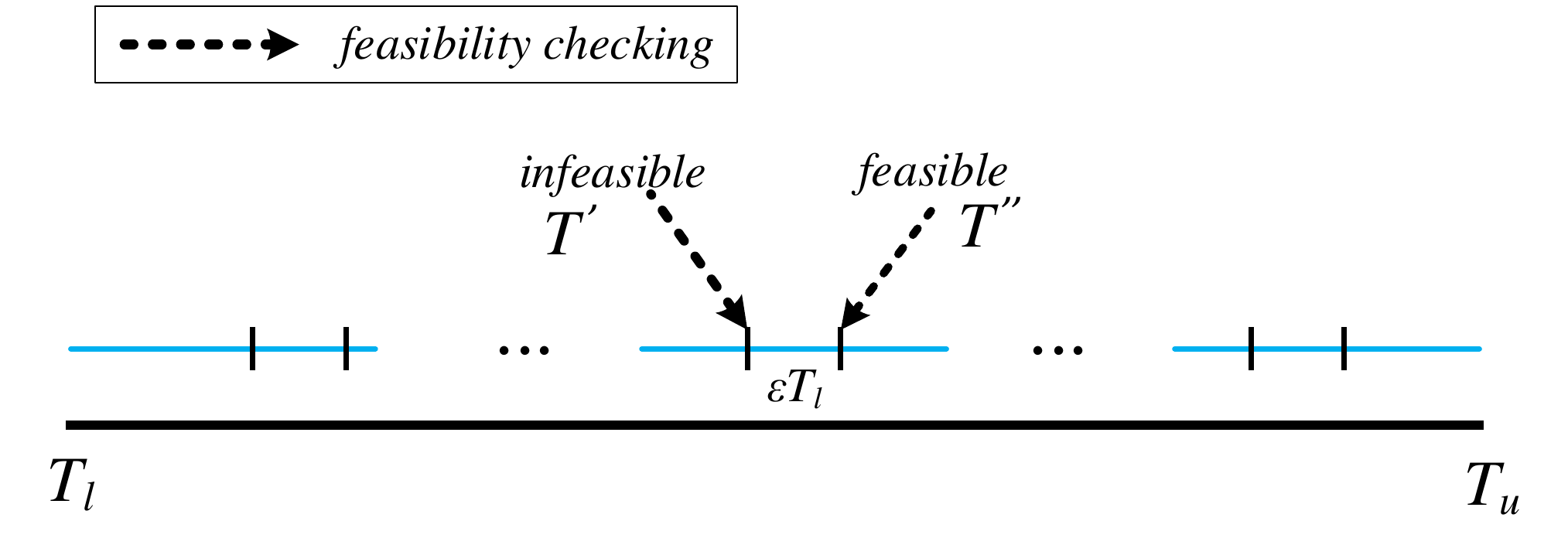}
    \caption{Binary search on $[T_l ,T_u]$ with accuracy level of $\epsilon \cdot T_l$.}
    \label{fig:binary_search}
\end{figure}


Thus, we can obtain the following fully polynomial-time approximation scheme (i.e., Algorithm~\ref{alg:fptas}) to solve problem (\ref{IOP}) by combining both binary search and Algorithm~\ref{alg:decalg}.

%
%
%

\begin{theorem}\label{thrm:FPTAS}
Let $T^*$ be the optimal deployment delay of problem (\ref{IOP}). Given any small allowable error $\epsilon >0$, there exists an FPTAS with running time $O(n^2 \log \frac{1}{\epsilon})$ to arbitrarily approach the global optimum (i.e., $T^* \leq T \leq (1 + \epsilon)T^*$).
\end{theorem}
\begin{proof}
The deployment delay of a given instance has an upper bounded of $T_u$ and a lower bound of $T_l$. Obviously, $T_l \leq T^* \leq T_u$. Choosing a small constant $\epsilon > 0$, we divide each $T_l$ into $\frac{1}{\epsilon}$ sub-intervals. Here, to make the discussion easier, we assume $\frac{1}{\epsilon}$ is an integer). Each interval has length $\epsilon \cdot T_l$, where $\epsilon \cdot T_l \leq \epsilon \cdot T^*$. We divide $T_u$ by $\epsilon \cdot T_l$ into $\lceil{\frac{T_u}{\epsilon \cdot T_l}}\rceil$ sub-intervals as $I$ as in Algorithm~\ref{alg:fptas}. Overall, we have $\lceil\frac{T_u}{\epsilon \cdot T_l}\rceil$ intervals on $I$.

Then, each step of binary search will shrink the interval $I$ by applying Algorithm~\ref{alg:decalg} on certain value of $T$. It terminates with deployment delays $T^{\prime}$ and $T^{\prime\prime}$, as shown in Figure~\ref{fig:binary_search}, in which $T^{\prime} < T^*$ and $T^{\prime\prime} = T^{\prime} + \epsilon \cdot T_l > T^*$. The resultant $T^{\prime\prime}$ is our searched optimum for problem (\ref{IOP}). We have that $T^{\prime\prime} = T^{\prime} + \epsilon \cdot T_l \leq T^* + \epsilon \cdot T_l \leq (1 + \epsilon) T^*$. Overall, we obtain $T^* < T \leq (1 + \epsilon) T^*$.

Therefore, we obtain the deployment delay which has an approximation ratio $1+\epsilon$ over the global optimum. Our feasibility checking algorithm runs in $O(n^2)$ time, and we have $O(\lceil\frac{T_u}{\epsilon \cdot T_l}\rceil)$ candidate deadlines. Overall, this algorithm runs in $O(n^2 \log \frac{1}{\epsilon})$ since binary search runs in at worst logarithmic time.
\end{proof}

\noindent Note that the relative error of the proposed FPTAS is only due to the small constant value $\epsilon$ that we choose in binary search.

\begin{algorithm}[t]
\caption{FPTAS for minimizing the maximum deployment delay}
\begin{algorithmic}[1]

\STATE \textbf{Input:}\\ 
$I = \{\epsilon T_l, 2 \epsilon T_l, \cdots,  \lceil{\frac{T_u}{\epsilon \cdot T_l}}\rceil \epsilon T_l\}$ where $T_l$ and $T_u$ are given in (\ref{eq:low}) and (\ref{eq:upp})\\

\STATE \textbf{Output:}\\ 
$I(idx)$: $idx$ is the index

\STATE $low \gets 1$ and $high \gets \lceil{\frac{T_u}{\epsilon \cdot T_l}}\rceil$
\WHILE{$low <= high$}
       \STATE {$mid \gets floor((low + high)/2)$}
       \STATE {feasibility checking by Algorithm~\ref{alg:decalg} on $I(mid)$}
       \IF {$I(mid)$ is feasible}
            \STATE {$high \gets mid$}
       \ELSE
            \STATE {$low \gets mid$}
       \ENDIF
        \IF {$low == high-1$}
            \STATE {$idx \gets high$}
            \STATE {break}
        \ENDIF
\ENDWHILE
\RETURN $I(idx)$
\end{algorithmic}
\label{alg:fptas}
\end{algorithm}

\section{Total UAV deployment delay optimization}\label{sec:sum}

In this section, we further consider the efficiency problem through minimizing the total UAV deployment delay for covering the target interval. We first show that problem (\ref{sum}) when UAV dispatching from different locations is NP-complete by reduction from {\em 3-partition problem} \cite{garey2002computers}. The proof is similar to Theorem~\ref{thry:np1}.

\begin{theorem} \label{thry:np2}
The total deployment delay minimization problem in (\ref{sum}) is NP-complete.
\end{theorem}
\begin{proof}
{See Appendix~\ref{ap:npc2}.}
\end{proof}

\subsection{Fast algorithm for UAVs deployment from the same location} \label{sec:samesum}

We first study the problem of dispatching the UAVs from the same initial location, i.e., $x_i = x_j$ for $1 \leq i ,j \leq n$. Different from the previous min-max optimization problem in Section~\ref{sec:same}, the problem here is still difficult to solve. Because in min-max optimization problem, we only focus on the bottleneck (the maximum one) of all UAVs' deployment delay and reduce deployment delays of all UAVs in a fair manner, while the min-sum problem targets at minimizing the sum of the deployment delays of selected UAVs in the solution. Without loss of generality, we assume that $x_i \leq 0$, $\forall \ \mu_i,  1 \leq i \leq n$, which is symmetrical to the case that $x_i \geq \beta$. 
Our problem of dispatching all UAVs simultaneously is equivalent to dispatching of UAVs one by one to cover the line interval $[0, \beta]$ from its right endpoint (or furthest point) $\beta$ to left endpoint (closest point) $0$. Intuitively, if all UAVs have the same flying speed and operating altitude, the optimal deployment scheme is to deploy a UAV with longer wireless coverage radius to further location for saving the travel distance and delay. Specifically, given a target interval ($[0,\beta]$ initially or remaining uncovered interval during deployment), we sequentially select the unused UAV with the longest wireless coverage radius among all available UAVs to reach the furthest point in the remaining uncovered interval. 

As shown in Algorithm~\ref{alg:commonsum}, initially, we set $\overline{\beta}= \beta$ and the available UAV set $\textbf{U}^-=\textbf{U}$. In each iteration, we dispatch a UAV $\mu_j$ from the available UAV set $\textbf{U}^-$ with longest wireless coverage radius $r_j$ to extend the current covered interval $[\overline{\beta}, \beta]$, as shown in Lines 7 and 8 of Algorithm~\ref{alg:commonsum}. Next, we add $\mu_j$'s travel time $T_j$ to $\textbf{T}$. Then, we update the covered interval to $[\overline{\beta}-2r_j, \beta]$ and remove the UAV $\mu_j$ from $\textbf{U}^-$ until the target interval is fully covered. In the end, we obtain the total deployment delay $\Gamma = \sum{\textbf{T}}$. Note that we may only select a subset of UAVs with minimum total deployment delay to cover the target interval $[0, \beta]$ in final solution, since $2 \sum_{i=1}^{n}{r_{i}} \geq \beta$.

\begin{lemma}\label{fast}
{If all UAVs have the same flying speed and operating altitude, Algorithm~\ref{alg:commonsum} optimally solves the min-sum deployment problem when dispatching $n$ UAVs from the same location.}
\end{lemma}
\begin{proof}
{See Appendix \ref{ap:fast}.}
\end{proof}

\begin{algorithm}[t]
\caption{Fast algorithm for dispatching UAVs from the same location}
\begin{algorithmic}[1]
\STATE \textbf{Input:}\\
$\textbf{U} =\{\mu_1,\mu_2, \ldots, \mu_n\}$ \\
\STATE \textbf{Output:}\\
$y_{i}$: final location of $\mu_{i}$ \\
$\Gamma$: total deployment delay

\STATE $\overline{\beta} = \beta, \textbf{U}^- = \textbf{U}, \textbf{T} = \emptyset$
\WHILE{$0 < \overline{\beta} \leq \beta$}
\STATE $\mu_j \gets \arg \max_{\mu_i \in \textbf{U}^-} r_i$
\STATE {$\textbf{T} \gets \textbf{T} \cup T_{j} = \frac{\sqrt{((\overline{\beta} - r_j)^2 + h_j^2)}}{v_j}$}
\STATE {$\overline{\beta} \gets \overline{\beta} - 2r_j$}
\STATE {$y_{j} \gets \overline{\beta} - r_j$ , \  $\textbf{U}^- \gets \textbf{U}^- \setminus \{\mu_j\}$}

\ENDWHILE
\RETURN {$\Gamma \gets \sum \textbf{T}$}
\end{algorithmic}
\label{alg:commonsum}
\end{algorithm}


\begin{proposition} \label{prop:samesum}
Let $\Gamma^*$ be the optimal total deployment delay of problem (\ref{sum}) when dispatching the UAVs from the same initial location. Algorithm~\ref{alg:commonsum} of computational complexity $O(n)$ can obtain the total deployment delay $\Gamma \leq \kappa \tau \Gamma^*$, where $\kappa = h_{max}/h_{min}$ and $\tau = v_{max}/v_{min}$.
\end{proposition}
\begin{proof}
Suppose that all the UAVs are fixed with the same altitude and flying speed, by applying Algorithm~\ref{alg:commonsum}, we can obtain the optimal solution of minimizing the total delay for dispatching UAVs from the same location by Lemma \ref{fast}.

Next, we assume that all the UAVs are fixed with the same altitude $h_{max} = \max h_i$ and flying speed $v_{min} = \min v_i$. Thus, on one hand, we can find the optimal solution by Algorithm~\ref{alg:commonsum} and obtain the total deployment delay $\Gamma_{max} = \frac{1}{v_{min}} \sum_{1 \leq i \leq n} \sqrt{y_i^2 + h_{max}^2}$. On the other hand, by applying the same algorithm, we find the total deployment delay $\Gamma_{min} = \frac{1}{v_{max}} \sum_{1 \leq i \leq n} \sqrt{y_i^2 + h_{min}^2}$ if all the UAVs are fixed with the same altitude $h_{min} = \min h_i$ and flying speed $v_{max} = \max v_i$. We can see that $\Gamma_{min} \leq \Gamma^*$, in which $\Gamma^*$ is the total deployment delay in the optimal solution. Moreover, since $\kappa = \frac{h_{max}}{h_{min}} \geq 1$ and $\tau = \frac{v_{max}}{v_{min}} \geq 1$, the following holds:

{\small
\setlength{\abovedisplayskip}{3pt}
\setlength{\belowdisplayskip}{1pt}
\begin{align*}
\Gamma_{max} &= \frac{1}{v_{min}} \sum \sqrt{y_i^2 + h_{max}^2} \\
  &           = \frac{\tau}{v_{max}} \sum \sqrt{y_i^2 + \kappa^2 h_{min}^2} \\
  &          \leq \kappa \tau \frac{1}{v_{max}} \sum \sqrt{y_i^2 +  h_{min}^2} \\
  &          \leq \kappa \tau \Gamma_{min}.
\end{align*}
}

The total deployment delay obtained by Algorithm~\ref{alg:commonsum} is $\Gamma$ and we have $\Gamma \leq \Gamma_{max}$. Thus, $\Gamma \leq \Gamma_{max} \leq \kappa \tau \Gamma_{min} \leq \kappa \tau \Gamma^*$.

With respect to the time complexity, we can see that there are at most $n$ iterations for the while loop, Algorithm~\ref{alg:commonsum} runs in linear time, which completes our proof.
\end{proof}
\noindent {Note that the minimum and maximum possible flying altitudes influence the computed total deployment delay. As the value of ratio $\frac{h_{max}}{h_{min}}$ (variance of flying altitudes) increases, the total deployment delay increases.} Algorithm~\ref{alg:commonsum} works in a greedy way based on the wireless coverage radius without considering the UAVs' diversity in operating altitude and flying speed. It has the advantage of low computational time. However, the gap between its obtained total deployment delay and the optimal one can be large if the variance of operating altitudes or flying speeds is large. To achieve a better performance, we can use the scheme with the pseudo-polynomial time algorithm developed in Section \ref{sec:dp}, which is designed for a more general setting of the min-sum problem.

\subsection{Reformulation of Problem (\ref{sum}) and bound analysis} \label{sec:differentsum}

We further study the general min-sum problem (\ref{sum}), when UAVs are dispatched from different locations. As in the min-max problem in Section~\ref{sec:iop}, here we add the order preserving constraint to make the analysis tractable. The problem is defined as follows:
{
\setlength{\abovedisplayskip}{2pt}
\setlength{\belowdisplayskip}{4pt}
\begin{align}
\label{IOPsum}
&~ \min_{\{y_1,\cdots, y_n\}} \underset{1 \leq i \leq n} {\sum} ~ {T}_i(y_i) ~ , \\
&~ {\rm s.t.}, \ [0,\beta] \subseteq  \bigcup_1^n [y_i - r_i, y_i + r_i],\notag \\
&~~~~~~~ \forall \ 1 \leq i \leq n-1, \  y_i \leq  y_{i+1}. \notag
\end{align}}
Note that the last inequality is due to the constraint of initial location order preserving given $x_1 \leq \cdots \leq x_n$.



It is still difficult to solve problem (\ref{IOPsum}) directly even under the constraint of order preserving, since there are still factorial number of combinations in solution. In the previous min-max optimization problem (\ref{IOP}), we use feasibility checking algorithm by assigning an identical deadline to all UAVs. However, it can not provide satisfactory solution for problem (\ref{IOPsum}), which targets at minimizing the sum of the deployment delays of selected UAVs in the solution. In spite of this, before presenting the optimal algorithm for problem (\ref{IOPsum}), we claim that we can still apply Algorithm~\ref{alg:fptas} for the new min-sum problem here to find a value $\Gamma^{\prime}$ that roughly approximates the optimal total deployment delay $\Gamma^*$. $\Gamma^{\prime}$ is the summation of all UAVs' delays obtained by Algorithm~\ref{alg:fptas}, which aims to minimize the maximum deployment delay. Next, we show the fact that the solution obtained by Algorithm~\ref{alg:fptas} can achieve an $n(1+\epsilon)$-approximation for problem (\ref{IOPsum}). Conversely, we can also show that the optimal solution of problem (\ref{IOPsum}) achieves an $n$-approximation for min-max problem (\ref{IOP}).



\begin{lemma}\label{bound}
$\Gamma^{*} \leq \Gamma^{\prime} \leq n(1+\epsilon) \Gamma^{*}$. Conversely, $T^* \leq T_{max}^* \leq n T^*$.
\end{lemma}
\begin{proof}
For any instance of problem (\ref{IOPsum}), we have $\Gamma^{*} = \sum_{1 \leq i \leq n} T_i^*$, and $T_{max}^*$ is the maximum one among all ${T_i^*}'s$ as shown in Table~\ref{tab:notation}. Then, we have $T_{max}^* \leq \Gamma^{*}$. As $T^*$ is the minimum maximum deployment delay of problem (\ref{IOP}), we have $T^* \leq T_{max}^*$ since the maximum deployment delay in problem (\ref{IOP}) is not lower than the maximum deployment delay in problem (\ref{IOPsum}). Since $T$ is obtained by Algorithm~\ref{alg:fptas}, the following holds:

{
\setlength{\abovedisplayskip}{3pt}
\begin{align*}
\Gamma^{\prime} & \leq  n \cdot T \leq n(1 + \epsilon) T^* \leq n(1 + \epsilon) T_{max}^* \leq n(1 + \epsilon) \Gamma^{*}
\end{align*}
}
Since $\Gamma^{*} \leq  \Gamma^{\prime}$, we conclude $\Gamma^{*} \leq \Gamma^{\prime} \leq n(1+\epsilon) \Gamma^{*}$.

Conversely, $T_{max}^*$ is the maximum deployment delay among all UAVs obtained by the optimal algorithm for problem (\ref{IOPsum}). Then, we have $T_{max}^* \leq \Gamma^{*}$, where $\Gamma^{*}$ is the optimal total deployment delay for problem (\ref{IOPsum}). We conclude $\Gamma^{*} \leq n \cdot T^*$, since the total deployment delay in problem (\ref{IOP}) is not lower than the total deployment delay in problem (\ref{IOPsum}). Thus, we obtain $T^* \leq T_{max}^* \leq n T^*$.
\end{proof}

\subsection{Dynamic programming for solving problem (\ref{IOPsum})} \label{sec:dp}

Different from the min-max optimization problem (\ref{IOP}), the feasibility checking algorithm by assigning an identical deadline to all UAVs can not provide satisfactory solution for problem (\ref{IOPsum}). Because problem (\ref{IOPsum}) is to compute the optimal configuration of the UAV network to coordinately minimize the sum of the deployment delays of selected UAVs in the configuration. Given the order $\propto$ defined in Section~\ref{sec:iop}, we present a dynamic programming approach for solving the problem (\ref{IOPsum}), which starts with the leftmost point in $[0,\beta]$ and sequentially dispatch the UAVs one by one according to $\propto$.

For the leftmost $i$ UAVs $\mu_1, \mu_2, \ldots, \mu_i$ and any given delay $j > 0$, we use $[0, R(i, j)]$ to denote the left-aligned interval covered by using only the leftmost $i$ UAVs within total deployment delay $j$. The initial value of $R(0, j) = 0$ and $R(i, 0) = 0$. If we want to cover the longest left-aligned interval with the leftmost $i$ UAVs (i.e., $\{\mu_{1}, \ldots, \mu_{i}\}$) and total deployment delay $j$, then we may or may not use UAV $\mu_{i}$. We are using the following recurrence to capture the idea that either the solution witnessing the left-aligned covered interval $R(i, j)$ uses $\mu_i$ and how much time $t$ is spent from $j$ in moving $\mu_i$ or else it does not ($R(i, j) = R(i-1, j)$).

If we do not use $\mu_i$, i.e., UAV $\mu_i$ can not be used to extend the current left-aligned covered interval within $T_i = t$, where $t$ is denoted as the time budget for UAV $\mu_i$. The longest left-aligned interval can be covered is $R(i, j) = R(i-1, j)$. In the other case where we do use $\mu_{i}$, the total deployment delay can be divided into two parts, i.e., $j-t$ and $t$, where $t$ is the delay of UAV $\mu_i$, and $j-t$ is the delay of the remaining $i-1$ UAVs. $t$ is feasible for $\mu_i$ if it allows $\mu_i$ to fly up vertically to $h_i$ at least, i.e., $(v_i t)^2 - h_i^2 \geq 0$. In the following, we use $\Delta_{(i,t)} > 0$ to denote $(v_i t)^2 - h_i^2$, then $\sqrt{\Delta_{(i,t)}}$ is the horizontal distance that UAV $\mu_i$ can move with delay $t$. By computing each time budget $t \in \{1, \ldots, j\}$ for moving UAV $\mu_i$, we select the one (best $t$ if it exists) that maximizing the left-aligned covered interval by using the leftmost $i$ UAVs $\mu_1, \mu_2, \ldots, \mu_i$ with total delay $j$. We have the following three cases that can possibly extend the currently covered left-aligned interval $R(i-1, j-t)$ depending on the relative initial position $x_i$ of $\mu_i$ and $R(i-1, j-t)$. Note that the currently covered left-aligned interval can not be extended in the cases that $x_i - r_i - \sqrt{\Delta_{(i,t)}} > R(i-1, j-t)$ and $x_i+r_i + \sqrt{\Delta_{(i,t)}} < R(i-1, j-t)$.

\begin{figure}[t]
    \centering
        \includegraphics[width=0.75\textwidth]{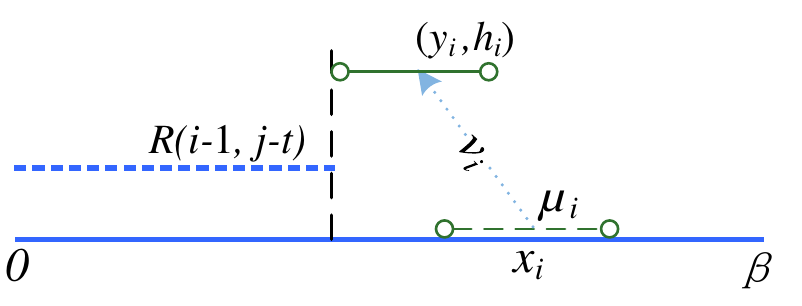}
    \caption{Case 2: deploying $\mu_i$ to the left to seamlessly cover from $R(i-1, j-t)$ where $x_i - r_i - \sqrt{\Delta_{(i,t)}} < R(i-1, j-t) < x_i - r_i$.}
    \label{fig:case2}
\end{figure}

\begin{figure}[t]
    \centering
        \includegraphics[width=0.75\textwidth]{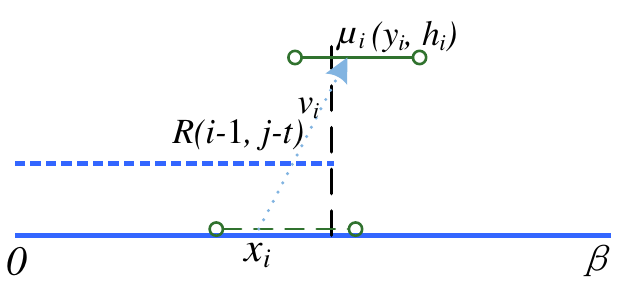}
    \caption{Case 3: deploying $\mu_i$ to the rightmost position where $x_i-r_i \leq R(i-1, j-t) < x_i+r_i + \sqrt{\Delta_{(i,t)}}$.}
    \label{fig:case3}
\end{figure}

\begin{itemize}
\item Case 1: If $\Delta_{(i, t)} < 0$, then we do not use UAV $\mu_i$ to cover the target interval. We have $R(i, j) = R(i-1, j)$. Otherwise, we have only the following two cases of using UAV $\mu_i$.
\item Case 2: If $x_i - r_i - \sqrt{\Delta_{(i,t)}} < R(i-1, j-t) < x_i - r_i$, as shown in Figure~\ref{fig:case2}, UAV $\mu_i$ can seamlessly cover from $R(i-1, j-t)$ and the new covered interval can be extended to $R^{\prime} = R(i-1, j-t) + 2r_i$. The new position of $\mu_i$ is $( y_i = R(i-1, j-t) + r_i, h_i)$.
\item Case 3: If $x_i-r_i \leq R(i-1, j-t) < x_i+r_i + \sqrt{\Delta_{(i,t)}}$, as shown in Figure~\ref{fig:case3}, the new covered interval can be extended to $R^{\prime} = \min \{ x_i+r_i + \sqrt{\Delta_{(i,t)}}, R(i-1, j-t) + 2r_i \}$. The new position of $\mu_i$ is $(R^{\prime} - r_i, h_i)$.
\end{itemize}

Moreover, if $R^{\prime} > R(i-1, j)$, then $R(i, j) = R^{\prime}$. Otherwise, $R(i, j) = R(i-1, j)$, i.e, $\mu_i$ will not be used. We can see that $R(i, j)$ is the longest left-aligned interval covered by the leftmost $i$ UAV within delay $j$.

The optimal total deployment delay for reaching full coverage of $L$ by using $n$ UAVs is as follows.
{
\setlength{\abovedisplayskip}{3pt}
\setlength{\belowdisplayskip}{1pt}
\begin{align}
\label{optdelay}
\Gamma^* = \min_{~\Gamma \geq 0}\{\Gamma~|~ R(n, \Gamma) \geq \beta \}.
\end{align}
}

The dynamic programming is given in Algorithm~\ref{alg:dpalg}, where we check the upper bound of $\Gamma$ (denoted as $\Gamma_u$) in problem (\ref{IOPsum}) to help search for the global optimum. For any UAV $\mu_i$, the maximum possible moving distance of $\mu_i$ is to reach the furthest position $(0, h_i)$ or $(\beta, h_i)$. Thus, $\Gamma$ to summarize all UAVs is loosely bounded by
{
\setlength{\abovedisplayskip}{3pt}
\setlength{\belowdisplayskip}{1pt}
\begin{align}
\label{eq:Tupp}
\Gamma_u = \sum_{1 \leq i \leq n} \{\frac{\max \{\sqrt{(\beta - x_i)^2 + h_i^2}, \sqrt{x_i^2 + h_i^2}\}}{v_i}\}.
\end{align}
}

\begin{algorithm}[t]
\caption{Dynamic programming for problem (\ref{IOPsum})}
\begin{algorithmic}[1]

\STATE \textbf{Input:}\\ 
$\textbf{U} =\{\mu_1,\mu_2, \ldots, \mu_n\}$\\

\STATE \textbf{Output:}\\ 
$\Gamma^{\prime\prime}$: total deployment delay

\FOR { $i = 0$ to $n$}
\FOR { $j = 0$ to $\Gamma_u$}
    \STATE {$R[i, j] \gets 0$}
\ENDFOR
\ENDFOR

\FOR { $i = 1$ to $n$}
\FOR { $j = 1$ to $\Gamma_u$}
\STATE {$R^{\prime} \gets 0$}
\FOR {$t = 1$ to $j$}
    \IF {$\Delta_{(i, t)} < 0$}
    \STATE {$R^{\prime} \gets R[i-1, j]$}
    \STATE {continue;}
        \ELSIF {$x_i - r_i - \sqrt{\Delta_{(i,t)}} < R(i-1, j-t) < x_i - r_i$}
	    \STATE {$R^{\prime} = R(i-1, j-t) + 2r_i$}
	    \STATE {$y_{i} \gets R(i-1, j-t) + r_i$}
            \ELSIF {$x_i-r_i \leq R(i-1, j-t) < x_i+r_i + \sqrt{\Delta_{(i,t)}}$}
	        \STATE {$R^{\prime} = \min \{ x_i+r_i + \sqrt{\Delta_{(i,t)}}, R(i-1, j-t) + 2r_i \}$}
	        \STATE {$y_{i} \gets R^{\prime} - r_i$}
                \ELSE
	            \STATE {$R^{\prime} \gets R[i-1, j]$}
    \ENDIF
    \ENDFOR
    \STATE {$R[i, j] \gets \max\{R[i-1, j], R^{\prime} \}$}
\ENDFOR
\ENDFOR
\RETURN $\Gamma^{\prime\prime} \gets \min\{j | 1 \leq j \leq  \Gamma_u, R(n,j) \geq \beta \}$\\
\end{algorithmic}
\label{alg:dpalg}
\end{algorithm}

The dynamic programming terminates with a table, whose $(i, j)$ entry records the value of $R(i, j)$. Each entry can be computed in constant time. To get the optimal solution, the whole table can be computed in $O(n \Gamma_u^2)$ time in worst case since $\Gamma \leq \Gamma_u$. Because $\Gamma_u$ may not be bounded by a polynomial of $n$, Algorithm~\ref{alg:dpalg} runs in pseudo-polynomial time.

\begin{theorem}\label{thrm:dppsudo}
Algorithm~\ref{alg:dpalg} returns the optimum of problem (\ref{IOPsum}) in pseudo-polynomial time.
\end{theorem}
\begin{proof}
We first show that the computed solution of total deployment delay $\Gamma^{\prime\prime}$ is feasible. We know that the final locations of UAVs in the computed solution follows order preserving. Moreover, the algorithm does not terminate until $R(n, \Gamma) \geq \beta$, then the target interval is fully covered. Thus, the solution of $\Gamma^{\prime\prime}$ output by Algorithm~\ref{alg:dpalg} is feasible.

It remains to show that $R(i, j)$ is the longest left-aligned interval covered by the leftmost $i$ UAVs within total delay $j$. Considering the optimal solution for $R(i, j)$, UAV $\mu_i$ is either dispatched or not. If not, then we have the same interval $R(i-1,j)$ covered by UAVs $\mu_{1},\ldots,\mu_{i-1}$. Alternatively, UAV $\mu_i$ is dispatched for $R(i, j) = R(i-1, j-t) + 2r_i$ or $R(i, j) = x_i+r_i + \sqrt{\Delta_{(i,t)}}$ as Figures~\ref{fig:case2} and~\ref{fig:case3}, respectively. In each case, $R(i, j)$ is maximized. 

There are three levels of $for$ loops in Line 8 ($n$ loops), Line 9 ($\Gamma_u$) and Line 11 ($\Gamma_u$). Thus, Algorithm~\ref{alg:dpalg} runs in time $O(n \Gamma_u^2)$, which is pseudo-polynomial time.
\end{proof}

By Lemma~\ref{bound} and Theorem \ref{thrm:dppsudo}, the Corollary \ref{clr} holds.
\begin{corollary}\label{clr}
Algorithm~\ref{alg:dpalg} can also attain an $n$-approximation for the min-max problem (\ref{IOP}).
\end{corollary}

\section{Performance Evaluation} \label{sec:exp}
In this section, we conduct simulations to evaluate the performances of our proposed UAV deployment algorithms. All the values reported later are collected from the average of 1000 runs for each algorithm. All UAVs are initially deployed randomly and the width $d$ is set as a small constant. The coverage radius, operating altitude and flying speed are randomly generated, while the length of target interval $L$, maximum flying speed, and maximum coverage radius are set to be 20 kilometers (km), 50 km per hour, and 3 km, unless otherwise stated. {We set the users' SNR threshold  $\gamma_{th}$ to be $10~dB$ as in \cite{mozaffari2015drone}, which determines the target data rate of each ground user (see Section \ref{sec:sys}).}

\begin{figure}[t]
    \centering
        \includegraphics[width=0.75\textwidth]{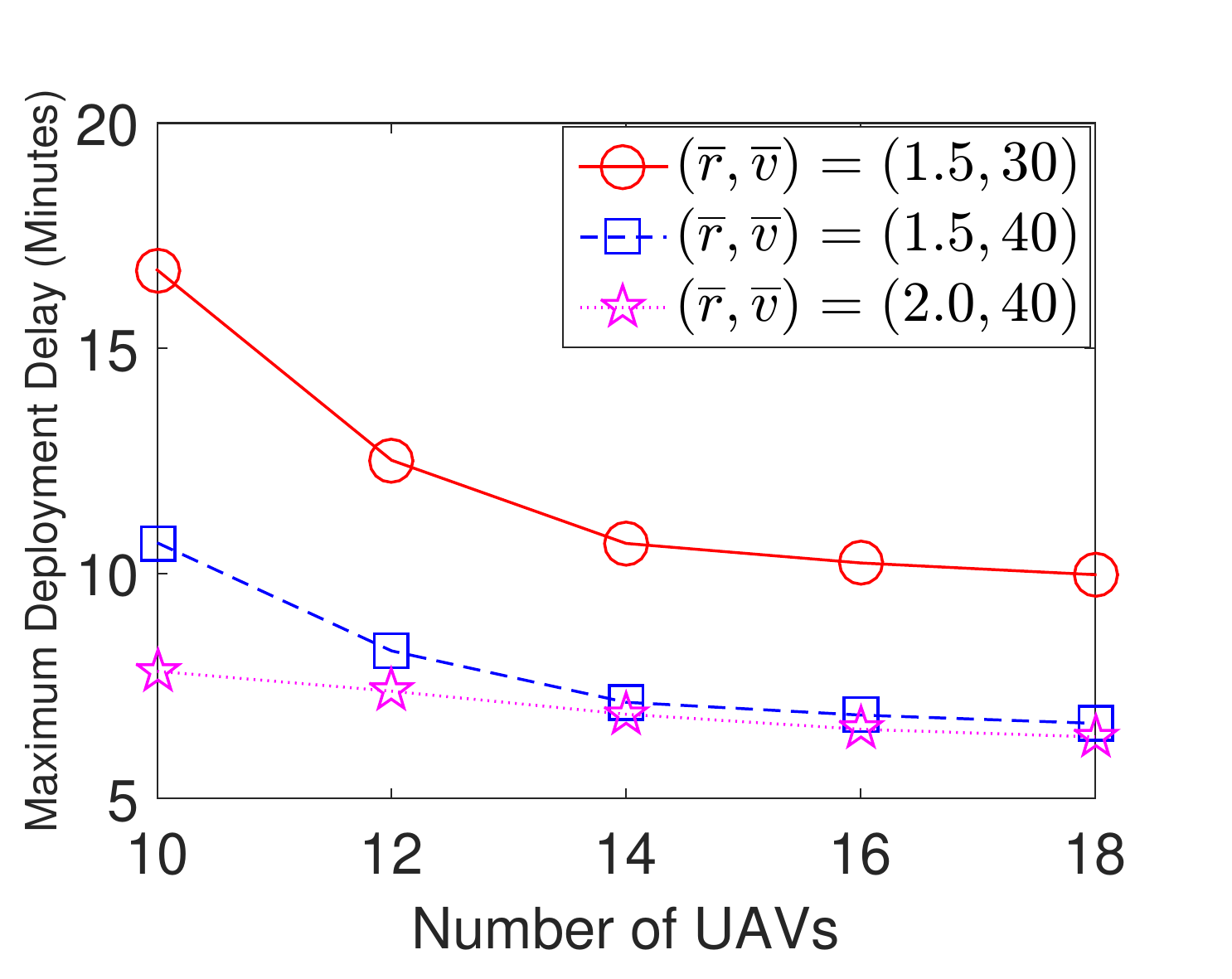}
    \caption{{The optimal deployment delay, the number of UAVs versus $\overline{r}$ (km) and $\overline{v}$ (km per hours) of all UAVs.}}
    \label{fig:same_mean}
\end{figure}

\subsection{Optimizing maximum deployment delay}\label{sec:globalexp}
In this section, we first present the experimental results for optimizing the maximum deployment delay to dispatch the UAVs in a fair manner.

\subsubsection{Dispatching of UAVs from the same location in problem (\ref{general})}\label{sec:expsame}

We first present the simulation results of Algorithm~\ref{alg:common} when dispatching the UAVs from the same location. Figure~\ref{fig:same_mean} shows the optimal deployment delay as a function of the number of UAVs under different mean values of coverage radius ($\overline{r} =\frac{1}{n}\sum_{i=1}^n r_i$) and flying speed ($\overline{v}=\frac{1}{n}\sum_{i=1}^n v_i$).

By increasing $\overline{v}$ or $\overline{r}$, the deployment delay decreases. Note that larger coverage radius (flying speed) of UAVs helps save the moving distance (time) to cover the whole target interval. By increasing the number of UAVs, the deployment delay decreases due to the increased UAV diversity and the flexibility to sample better UAVs. Still, there is a converging trend of the deployment delay with the increase of UAV number. Therefore, depending on the size of the target area and potential size of UAVs, an appropriate number of UAVs needs to be selected and deployed.

\subsubsection{Dispatching of UAVs from different locations in problem (\ref{IOP})}\label{sec:expdif}

\begin{figure}[t]
    \centering
    \includegraphics[width=0.75\textwidth]{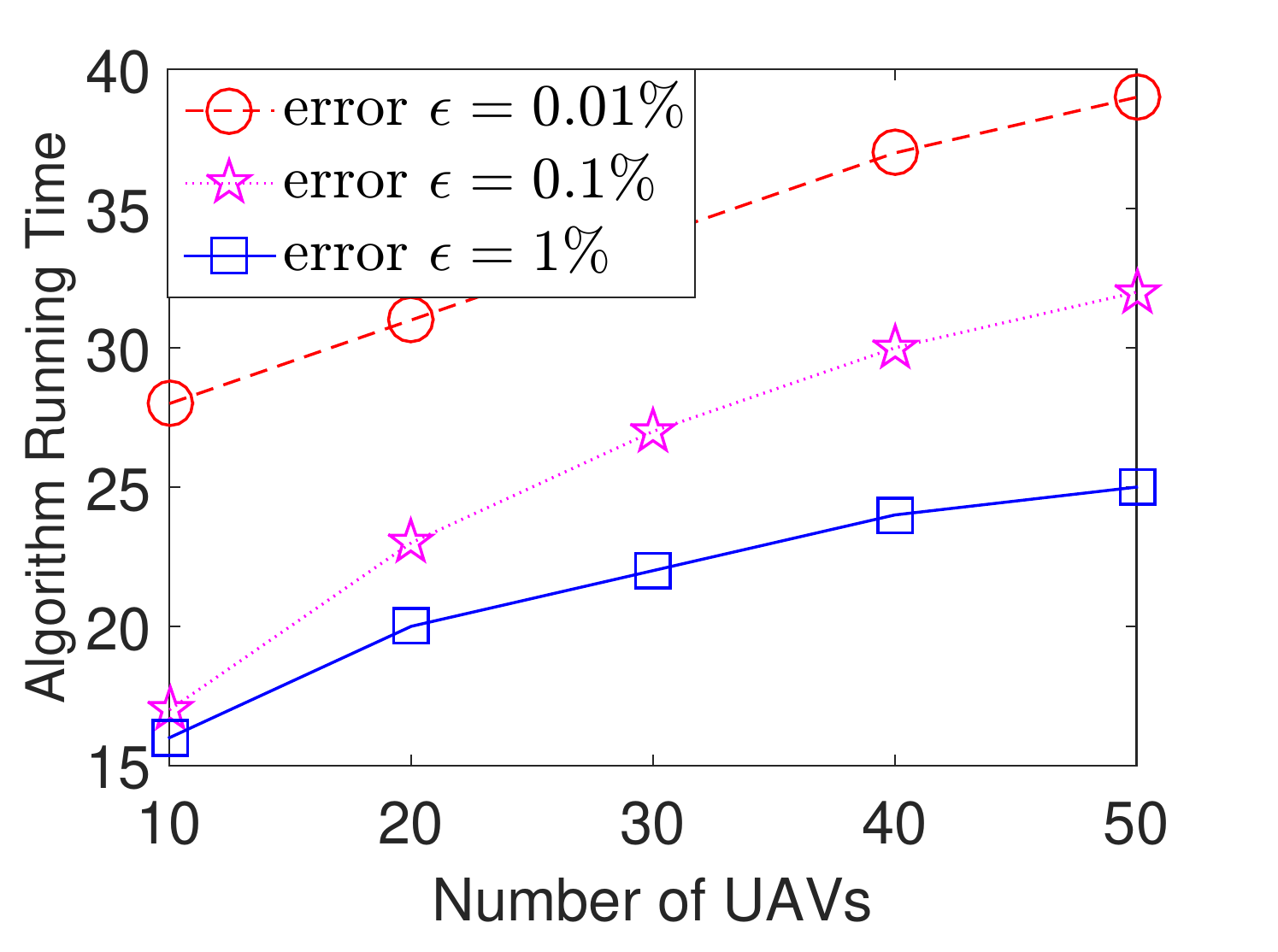}
    \caption{The running time (in milliseconds) of our approximation algorithms with different values of relative error $\epsilon$.}
\label{fig:op_alg_time}
\end{figure}

By running FPTAS in Algorithm \ref{alg:fptas}, we first show the time complexity for solving problem (\ref{IOP}). In Figure~\ref{fig:op_alg_time}, we show the running time of our approximation algorithms under different values of $\epsilon$, i.e., $1\%$, $0.1\%$ and $0.01\%$. It is observed that the smaller value of $\epsilon$ is, the larger running time is required. In addition, as the number of UAVs increases, the running time is concavely increasing, which is actually much smaller than the theoretical bound $O(n^2)$ in Theorem \ref{thrm:FPTAS}. This is because as the increase of the number of UAVs $n$ while the length of the line interval is fixed, Algorithm~\ref{alg:decalg} may not need to compute all the UAVs in Line 5.

Figure~\ref{fig:iop} shows the difference of the maximum deployment delay between problem (\ref{IOP}) and the original problem (\ref{general}). Actually, it examines the performance loss due to the proposal of preserving UAVs' initial locations for Algorithm \ref{alg:fptas}. In this figure, problem (\ref{IOP}) is solved by our proposed FPTAS by setting $\epsilon = 0.1\%$ and $0.001\%$, while problem (\ref{general}) is solved optimally by Brute-Force algorithm despite the high complexity. The generated flying speed and coverage radius are uniformly distributed, and the minimum coverage radius is set to be 4 kilometers to guarantee the full coverage. It can be observed that our proposed FPTAS for solving the reformulated problem (\ref{IOP}) can obtain a close deployment delay when comparing to the optimal deployment delay (obtained by Brute-Force algorithm for problem (\ref{general})). The gap does not necessarily increase with the number of UAVs, as our FPTAS greatly benefits from more UAVs.

\begin{figure}[t]
    \centering
    \includegraphics[width=0.75\textwidth]{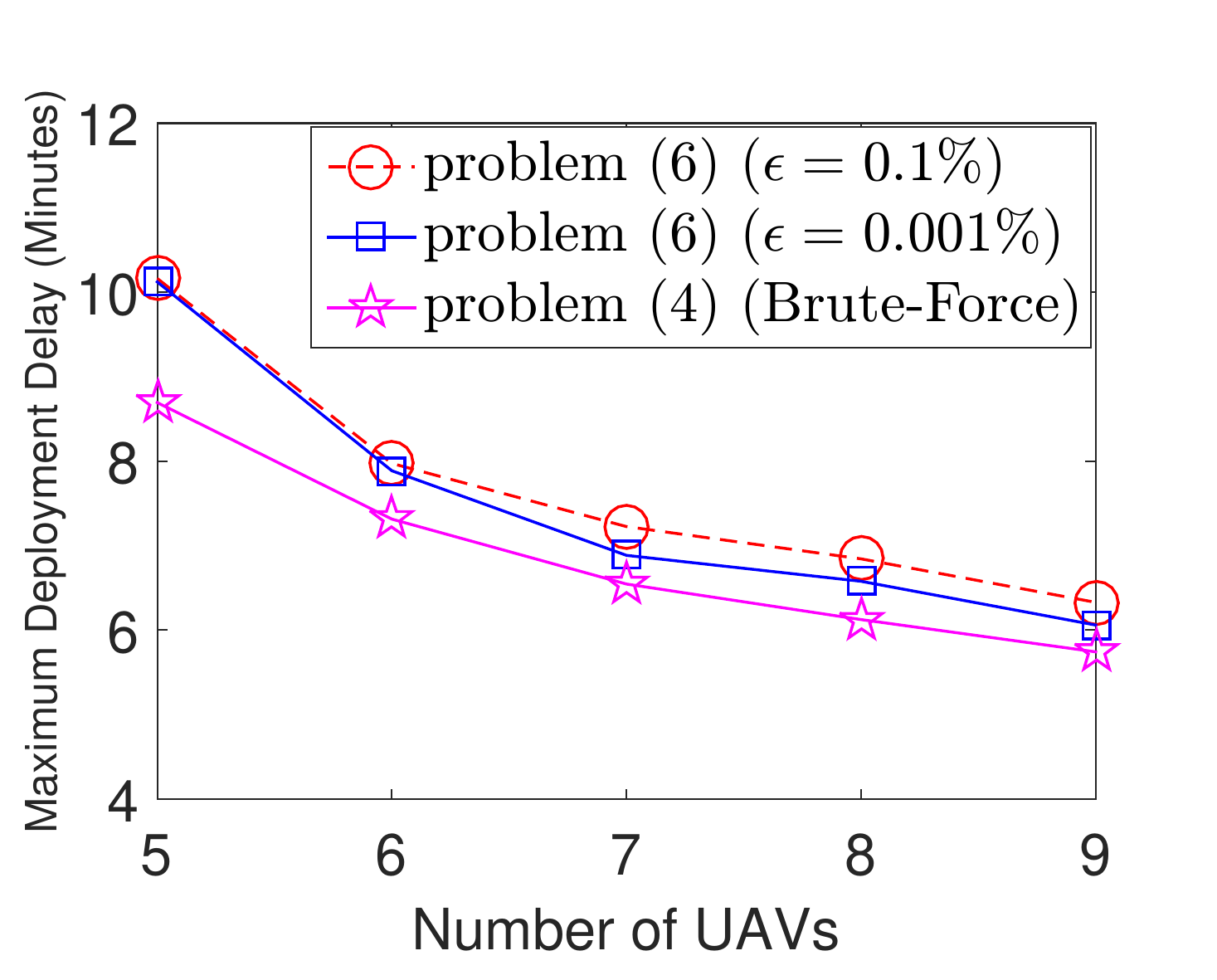}
    \caption{The deployment delays obtained by FPTAS for problem (\ref{IOP}) and Brute-Force algorithm for original problem (\ref{general}).}
    \label{fig:iop}
\end{figure}

\subsection{Optimizing total deployment delay}\label{sec:overallexp}
In this section, we further present the evaluation on algorithms for efficiently optimizing the total deployment delay for covering the target interval.

\begin{figure}[t]
    \centering
        \includegraphics[width=0.75\textwidth]{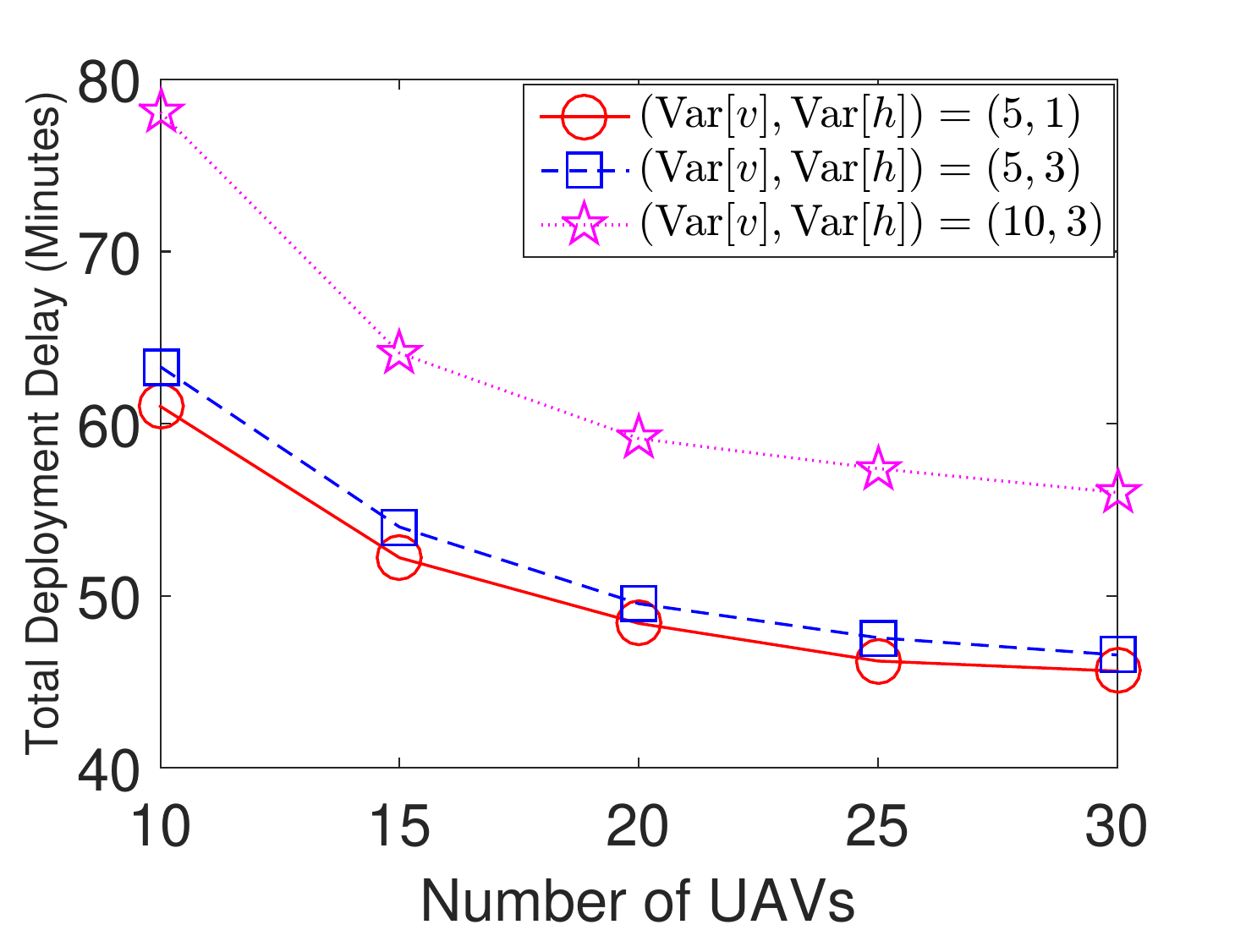}
    \caption{{The optimal deployment delay versus the number of UAVs versus $\mathrm{Var}[v]$ and $\mathrm{Var}[h]$ of all UAVs.}}
    \label{fig:same_overall}
\end{figure}

\subsubsection{Dispatching of UAVs from the same location in problem (\ref{sum})}\label{sec:overallsame}
We present the simulation results of Algorithm~\ref{alg:commonsum} when dispatching the UAVs from the same location. In this experiment, we set mean $\overline{v}$ as 40 km per hour, $\overline{r}$ as 2.0 kilometer, $\overline{h}$ as 5.0 kilometer. Figure~\ref{fig:same_overall} shows the optimal the total deployment delay as a function of the number of UAVs under different variances of flying speed $\mathrm{Var}[v]$ and operating altitude $\mathrm{Var}[h]$.

By increasing $\mathrm{Var}[v]$ or $\mathrm{Var}[h]$, the deployment delay increases, because $\kappa = h_{max}/h_{min}$ and $\tau = v_{max}/v_{min}$ become larger. This is consistent with Proposition \ref{prop:samesum}. By increasing the number of UAVs, the deployment delay decreases due to the increased UAV diversity and the flexibility to sample more appropriate UAVs. The influence of variance of $h_i$ is relative minor compared to variance of $v_i$ since the final moving distance is determined by both horizontal distance and operating altitude. Still, there is a converging trend of the deployment delay with the increase of UAV number.

\subsubsection{Dispatching of UAVs from different locations in problem (\ref{IOPsum})}\label{sec:overalldiff}

Similar to the results in Figure~\ref{fig:iop}, we can show that Algorithm~\ref{alg:dpalg} introduces only small performance loss due to the constraint of preserving UAVs' initial locations. Next we compare the performance between Algorithm~\ref{alg:fptas} (providing $n(1+\epsilon)$-approximation in Lemma~\ref{bound}) and Algorithm~\ref{alg:dpalg} (optimal for the min-sum problem) for total deployment delay minimization problem (\ref{IOPsum}).

\begin{figure}[t]
    \centering
        \includegraphics[width=0.75\textwidth]{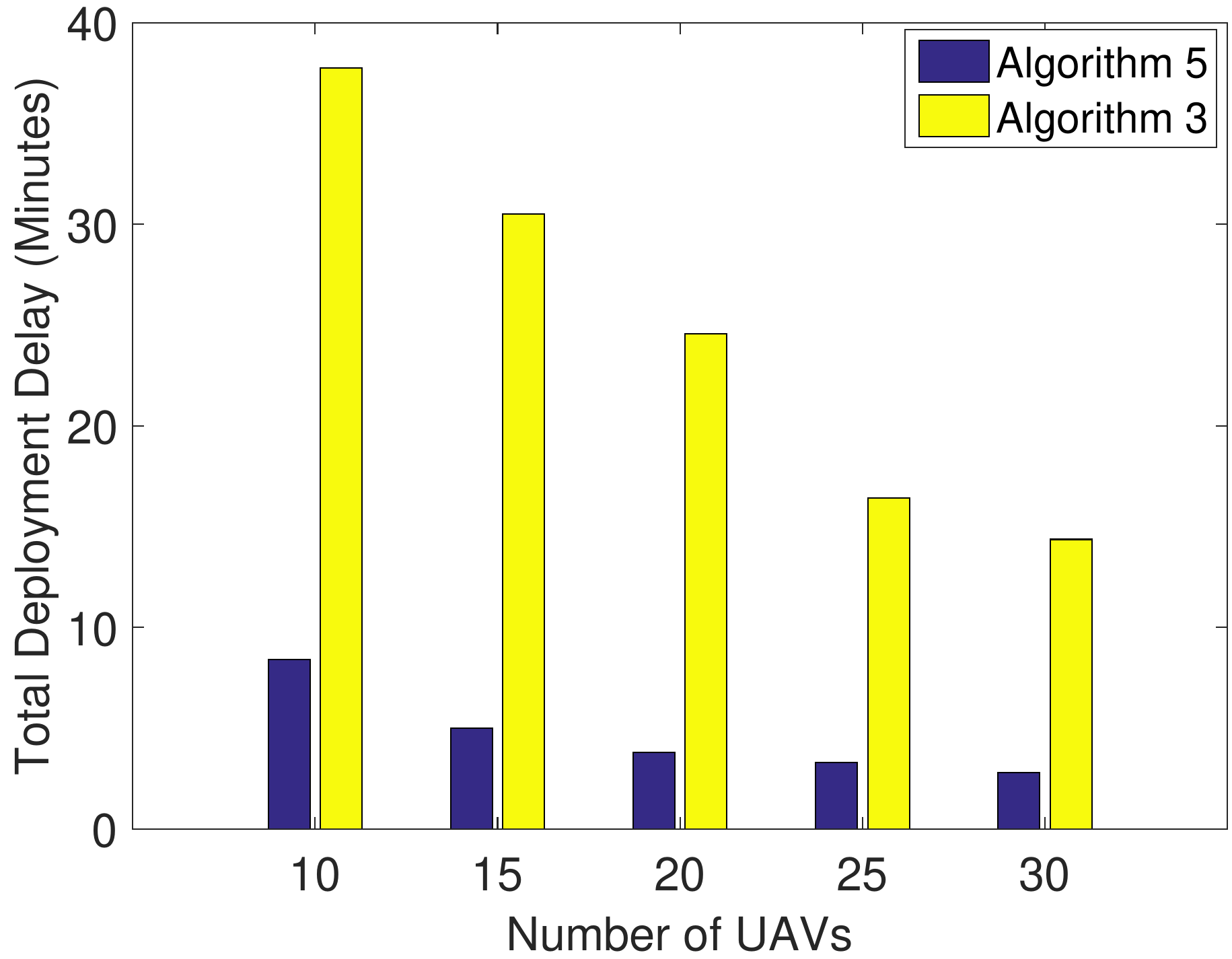}
    \caption{{The total deployment delay comparison between Algorithm~\ref{alg:fptas} and Algorithm~\ref{alg:dpalg}.}}
    \label{fig:dif_overall30}
\end{figure}

In Figure~\ref{fig:dif_overall30}, problem (\ref{IOPsum}) is solved by both Algorithm~\ref{alg:dpalg} and Algorithm~\ref{alg:fptas} with $\epsilon = 0.01\%$. Since Algorithm~\ref{alg:dpalg} provides the optimal solutions for the min-sum design purpose, it always obtains lower total deployment delays than Algorithm~\ref{alg:fptas} (for min-max design purpose) in Figure~\ref{fig:dif_overall30}. However, Algorithm~\ref{alg:dpalg} (pseudo-polynomial time) needs more computational time than Algorithm~\ref{alg:fptas} (in $O(n^2 \log \frac{1}{\epsilon})$). By increasing the number of UAVs, the deployment delays obtained by both algorithms decreases due to the UAV diversity gain. Figure~\ref{fig:dif_overall30} tells that minimizing the maximum deployment delay can imply a significant increase in total deployment delay. However, the empirical performance of Algorithm~\ref{alg:fptas} is better than the worst-case theoretical upper bounds indicated in Lemma~\ref{bound}.

\begin{figure}[t]
    \centering
        \includegraphics[width=0.75\textwidth]{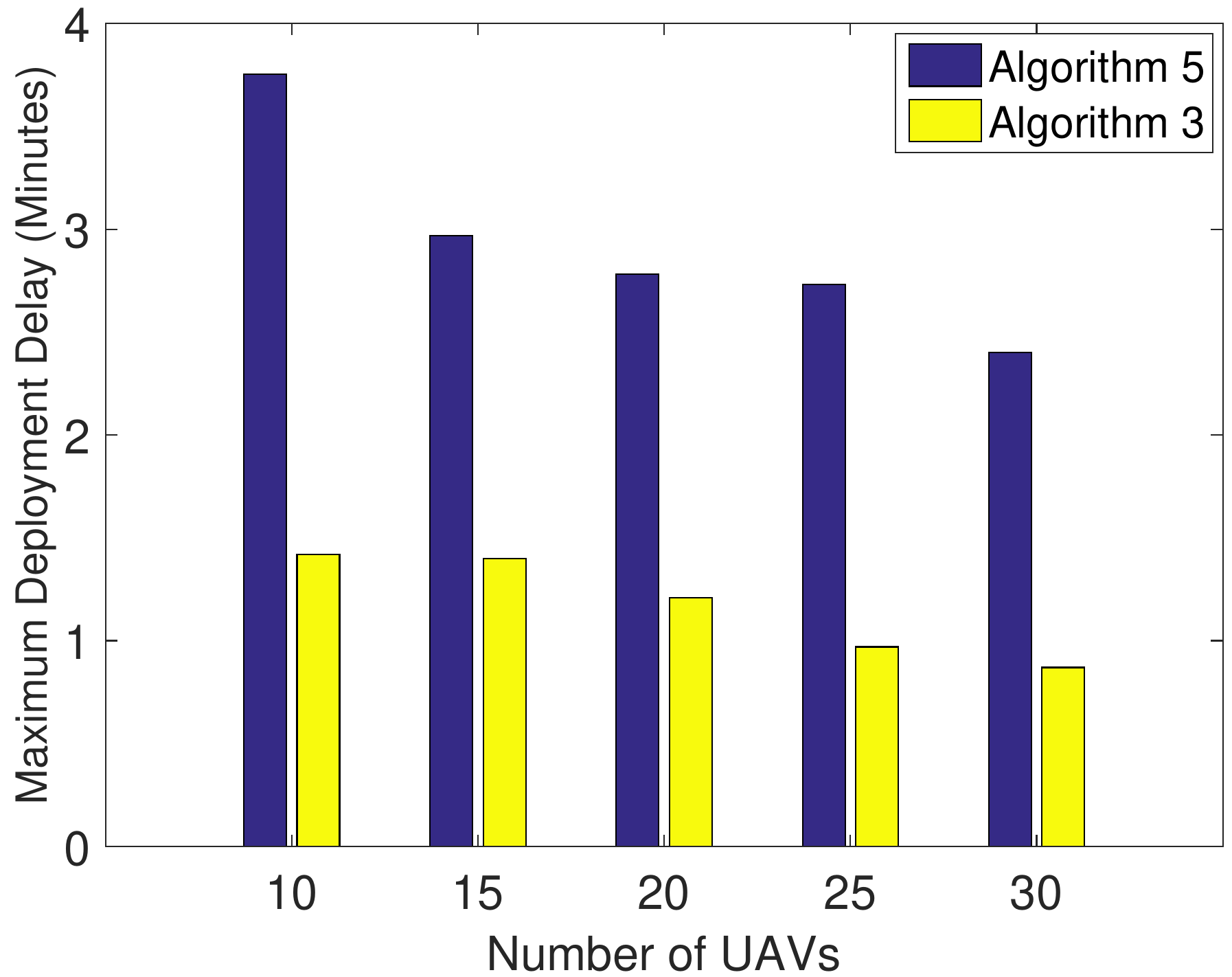}
    \caption{{Comparison between Algorithm~\ref{alg:fptas} and Algorithm~\ref{alg:dpalg} in terms of the maximum deployment delay.}}
    \label{fig:dif_mm30}
\end{figure}

Similarly, in Figure~\ref{fig:dif_mm30}, min-max problem~(\ref{IOP}) is solved by both Algorithm~\ref{alg:dpalg} ($n$-approximation according to Corollary~\ref{clr}) and Algorithm~\ref{alg:fptas} (providing $(1 + \epsilon)$-approximation, $\epsilon = 0.01\%$ here in the simulation). We show the performances of Algorithm~\ref{alg:fptas} and Algorithm~\ref{alg:dpalg} in terms of maximum deployment delay. It can be observed that the maximum deployment delay obtained by Algorithm~\ref{alg:fptas} is lower than Algorithm~\ref{alg:dpalg}. Figure~\ref{fig:dif_mm30} also tells that minimizing the total deployment delay can imply a significant increase in maximum deployment delay. However in fact, the empirical performance of Algorithm~\ref{alg:dpalg} not as bad as the worst-case theoretical upper bound indicated in Corollary~\ref{clr}.

\section{2D Extension for deployment algorithms} \label{sec:2D}

In this section, we discuss how to extend the proposed algorithms with theoretical guarantee to 2D. We consider to relax our model in two perspectives. One is to relax the initial locations of UAVs, the other is to relax the target area. Due to page limit, we only study the generalized min-max problem and the other generalized min-sum problem can be analyzed similarly.

\begin{figure}[t]
    \centering
        \includegraphics[width=0.75\textwidth]{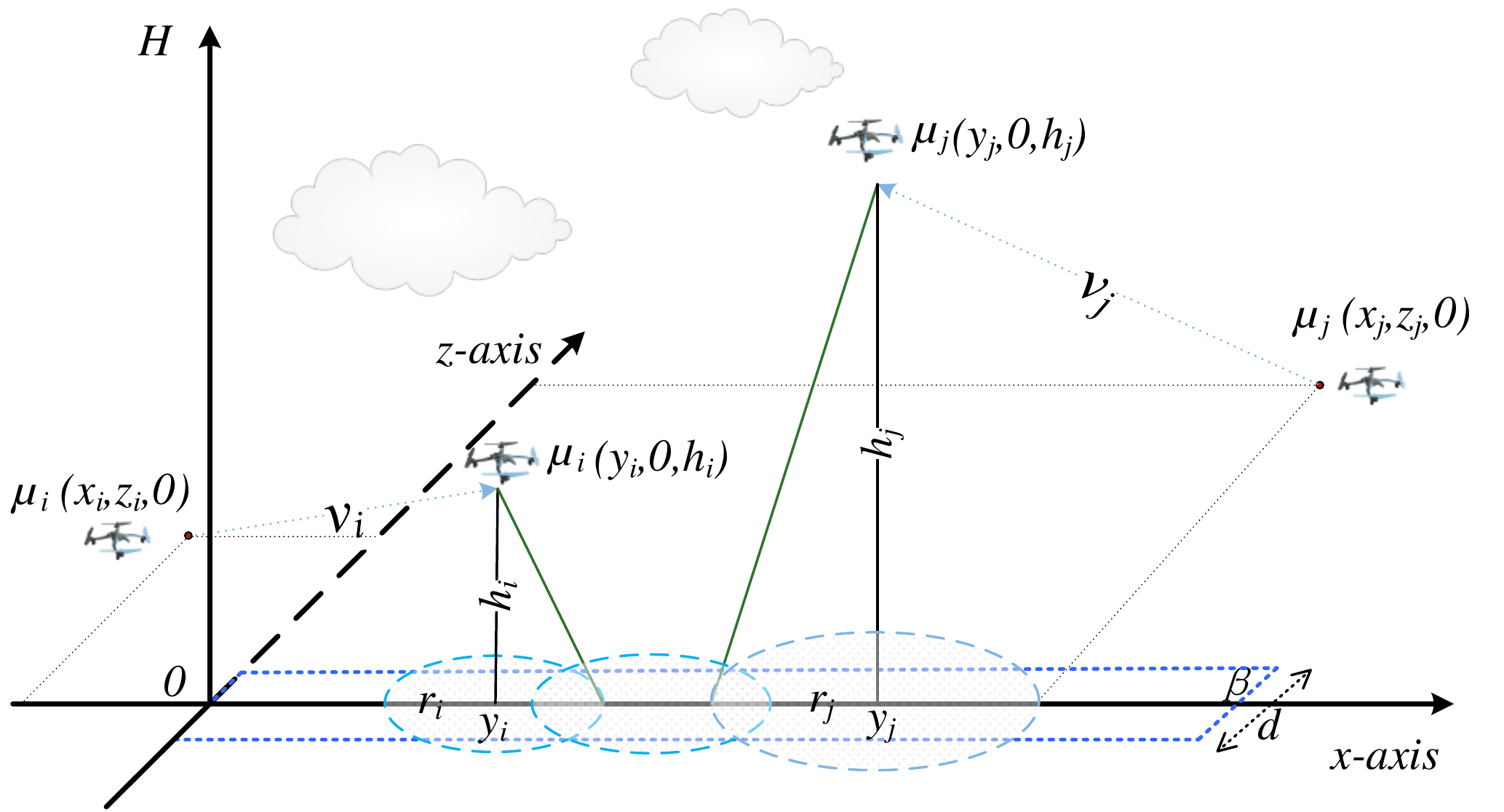}
    \caption{{Adding one dimension of the UAVs' initial positions, i.e., $z$-axis. $x$-domain and $z$-domain constitute the 2D ground space.}}
    \label{fig:covermodel1}
\end{figure}

\subsection{{UAVs are initially located in 2D area}} \label{sec:uav2D}

In current model, all UAVs are initially located on a line interval, i.e., $x$-axis. It can be extended to a 2D area by adding one more dimension of the UAVs' initial positions, i.e., $z$-axis. Specifically, for each UAV $\mu_i$, we use $z_i$ to denote the UAV $\mu_i$'s offset along $x$-axis. Thus, the initial location of the UAV is ($x_i, z_i, 0$), and it will be deployed to ($y_i, 0, h_i$) cover the target region, as shown in Figure \ref{fig:covermodel1}.
During the deployment, UAV $\mu_i$ travels an Euclidean distance $\sqrt{(y_i-x_i)^2 + z_i^2 + h_i^2}$ at flying speed $v_i$. Thus, its travel time is given by
{
\begin{equation}
\label{eq:time1}
{T}_i(y_i) = \frac{\sqrt{(y_i - x_i)^2 + z_i^2 + h_i^2}}{v_i}.
\end{equation}}

A particular UAV $\mu_i$ hovering at final position $(y_i, 0, h_i)$ covers a region $D_i$ (Equation (\ref{eq:area})) in the target area. It is required that all points of the target area \textbf{A} are covered after the deployment of UAVs, i.e., $\textbf{A} \subseteq \bigcup_1^n D_i$. We only need to change the travel time function ${T}_i(y_i)$ in Algorithm \ref{alg:common} and Algorithm \ref{alg:fptas} to be as in Equation (\ref{eq:time1}). Both algorithms can be applied directly when all UAVs are initially located in 2D area, and all theories still hold by similar proofs in Proposition \ref{common} and Theorem \ref{thrm:FPTAS}.


\subsection{{Both UAVs and target area are in 2D}} \label{sec:area2D}

In this section, we further relax our model to fast deploy diverse UAVs to provide full wireless coverage over a 2D rectangular ground plane $[0, \beta] \times [0, d]$.

Recall that both min-max and min-max problems are NP-complete when the target is line interval in Theorem \ref{thry:np1} and \ref{thry:np2}, thus the general problem of covering a 2D area is also NP-complete. Moreover, to reach a full coverage in the plane by using UAVs' disk-shaped coverage circles (no longer line segments in the 1D model) is general very difficult to solve \cite{stephenson2005introduction}, because the full coverage problem without any interstices is difficult to be solved by using UAVs' non-uniform coverage circles. Even if we have a solution for this static circle packing, we cannot use it for our fast UAV deployment problem because we also aim to minimize the travel time during the deployment. Despite the difficulties above, we manage to extend our prior algorithms to 2D area by applying proper approximations.

\subsubsection{{Uniform coverage radius}}
We first look at the case that all UAVs have the same coverage radius, then grid the rectangular area so that each grid square can be covered by a UAV \footnote{Since the coverage radius of UAVs is much smaller than the target area, we assume the area are decomposed into integral number of squares.}.

\begin{figure}[t]
    \centering
        \includegraphics[width=0.75\textwidth]{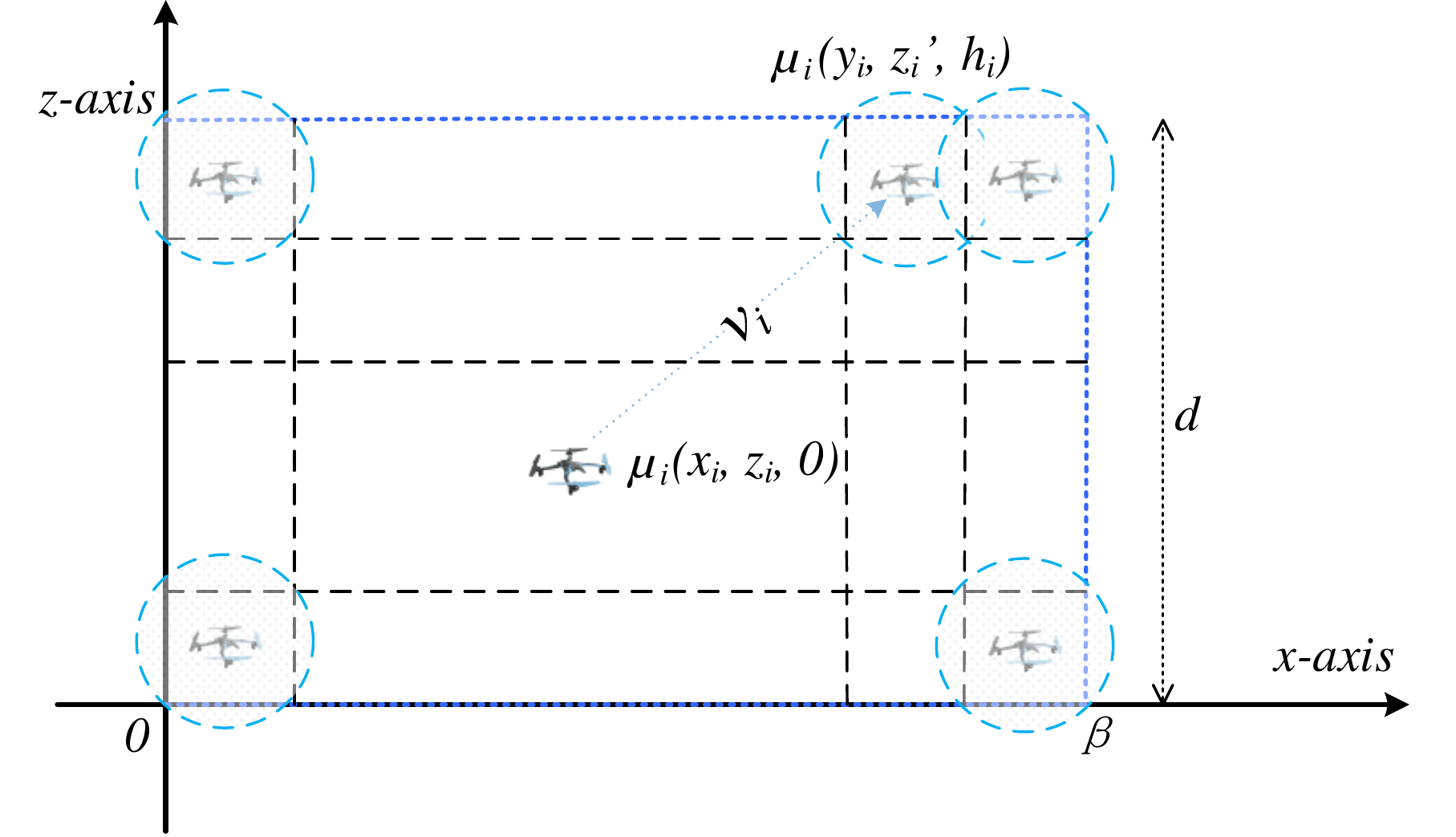}
    \caption{{Deploying UAVs to provide full wireless coverage over the rectangular area \textbf{A}, where the coverage radius of UAVs is much smaller than the target area.}}
    \label{fig:covermodel2}
\end{figure}

As shown in Figure~\ref{fig:covermodel2}, we want to fast deploy UAVs to provide full wireless coverage over the rectangular area \textbf{A} in equation (\ref{eq:A}). A particular UAV $\mu_i$ operating at final position $(y_i, h_i, z_i^{\prime})$ covers a region $A_i$ as equation (\ref{eq:Acover}) in the target rectangle. We require a full coverage over a rectangular area \textbf{A} by deploying $n$ diverse UAVs with identical coverage radius $r$.
{
\begin{equation}\label{eq:Acover}
\begin{split}
A_i = \{ (w,l)|&z_i^{\prime} -\frac{\sqrt{2}}{2} r_i \leq w \leq z_i^{\prime} + \frac{\sqrt{2}}{2} r_i,\\
&y_i -\frac{\sqrt{2}}{2} r_i \leq l \leq y_i + \frac{\sqrt{2}}{2} r_i \}.
\end{split}
\end{equation}}

We first study the min-max problem (\ref{general}) of covering a 2D area as in Section \ref{sec:same}, when UAVs are dispatched from the same location, i.e., ($x_i=0, z_i=0$) for all UAVs. Algorithm \ref{alg:common} can be applied directly. Since any UAV with larger distance from the initial location to the target position, it needs a larger travel time. Among all UAVs, we first consider which UAV to send and cover the furthest square of the target area. Specifically, given the current uncovered area, we sequentially select an unassigned UAV (e.g., $\mu_i$) with the minimum travel time to just cover the furthest square on the remaining uncovered area during deployment. We only need to compare each UAV's travel time to calculate the maximum delay objective.

We further study the min-max problem (\ref{IOP}) of covering a 2D area, when UAVs are dispatched from different locations. The new problem can be generalized from problem (\ref{IOP}):
{
\begin{align}
\label{IOP2D}
&~ \min_{\{(y_1,z_1^{\prime}),\cdots, (y_n,z_n^{\prime}\}} \underset{1 \leq i \leq n} {\max} ~ {T}_i(y_i, z_i^{\prime})~, \\
&~ {\rm s.t.}, \ \textbf{A} \subseteq \bigcup_1^n A_i,\notag \\
&~~~~~~~~~ \forall \ 1 \leq i \leq n-1, \  y_i \leq  y_{i+1}, \notag\\
&~~~~~~~~~ \forall \ 1 \leq i \leq n-1, \  z_i^{\prime} \leq  z_{i+1}^{\prime}. \notag
\end{align}}
{Note that the last two inequalities denotes the constraint of initial location order preserving along $x$-axis and $z$-axis for possible collision avoidance.}

{We can decompose the problem (\ref{IOP2D}) into $p$ subproblems, in which each subproblem $P_{i}$ are given a set of sequential UAVs $\Phi(t_i,\lambda_i) = \{ \mu_{t_i}, \mu_{t_i+1}, \cdots, \mu_{t_i + \lambda_i - 1}\}$, then UAVs in $\Phi(t_i,\lambda_i)$ are assigned to cover $q$ squares and they are with the same $x$-coordinate.}

\begin{algorithm}[t]
\caption{{Feasibility checking algorithm for 2D UAV deployment}}
\begin{algorithmic}[1]

\STATE \textbf{Input:}\\ 
$\textbf{U} =\{\mu_1, \mu_2, \ldots, \mu_n\}$ \\
$T$: a given deployment delay deadline for all UAVs

\STATE \textbf{Output:}\\ 
$y_{i}, z_i^{\prime}$: final locations of $\mu_{i}$

\STATE Compute $p = \left \lceil \frac{\beta}{\sqrt{2} r} \right \rceil$ and $q = \left \lceil \frac{d}{\sqrt{2} r}\right \rceil$ \\
\COMMENT{Calculate the number of UAVs needed to cover the 2D area.}
\STATE $t_{0}=1, \lambda_{0} = 0$
\FOR { $i = 1$ to $p$}
\FOR { $j = q$ to $n-t_{i-1}-\lambda_{i-1}-(p-i)q$ }
\STATE {Apply Algorithm \ref{alg:decalg} to subproblem $P(i)$ with $\Phi(t_{i-1}+\lambda_{i-1},j)$} \\
\COMMENT{The remaining UAVs are insufficient to cover the residual area.}
\IF {$T$ is feasible for $P(i)$}
	\STATE {$\lambda_i \gets j$, $t_i = t_{i-1}+\lambda_{i-1}$}
\ELSE
	\STATE {continue;}
\ENDIF
\IF {$j == n-t_{i-1}-\lambda_{i-1}-(p-i)q$}
    \RETURN $T \ is \ not feasible$ for problem (\ref{IOP2D})
\ENDIF
\ENDFOR
\ENDFOR
\RETURN $T \ is \ feasible$ for problem (\ref{IOP2D})\\

\end{algorithmic}
\label{alg:2D}
\end{algorithm}

{By combining Algorithm \ref{alg:2D} and binary search (similar to Algorithm \ref{alg:fptas}), we can obtain an FPTAS (\emph{2D deployment Algorithm}) to solve problem (\ref{IOP2D}).}

\begin{proposition} \label{prop:area2D}
{2D deployment Algorithm runs in $O(n^3 \log \frac{1}{\epsilon})$, which can arbitrarily approach the global optimum by assuming the UAVs' coverage radii are identical.}
\end{proposition}
\begin{proof}
Similar to Equation (\ref{eq:low}) and (\ref{eq:upp}), we can find the lower and upper bounds of the delay $T$. Then, we use binary search over those feasible deadlines to find the minimum deployment delay $T \leq (1 + \epsilon) T^*$ as in Section \ref{sec:bs}.
With respect to the time complexity, we can see that there are at most $p \cdot \frac{n}{q}$ iterations for the for loops, and Algorithm \ref{alg:decalg} runs in $O(n^2)$ time. Overall, Algorithm~\ref{alg:2D} runs in $O(n^3)$ time, which implies the obtained FPTAS runs in $O(n^3 \log \frac{1}{\epsilon})$.
\end{proof}

\begin{figure}[t]
    \centering
    \includegraphics[width=0.75\textwidth]{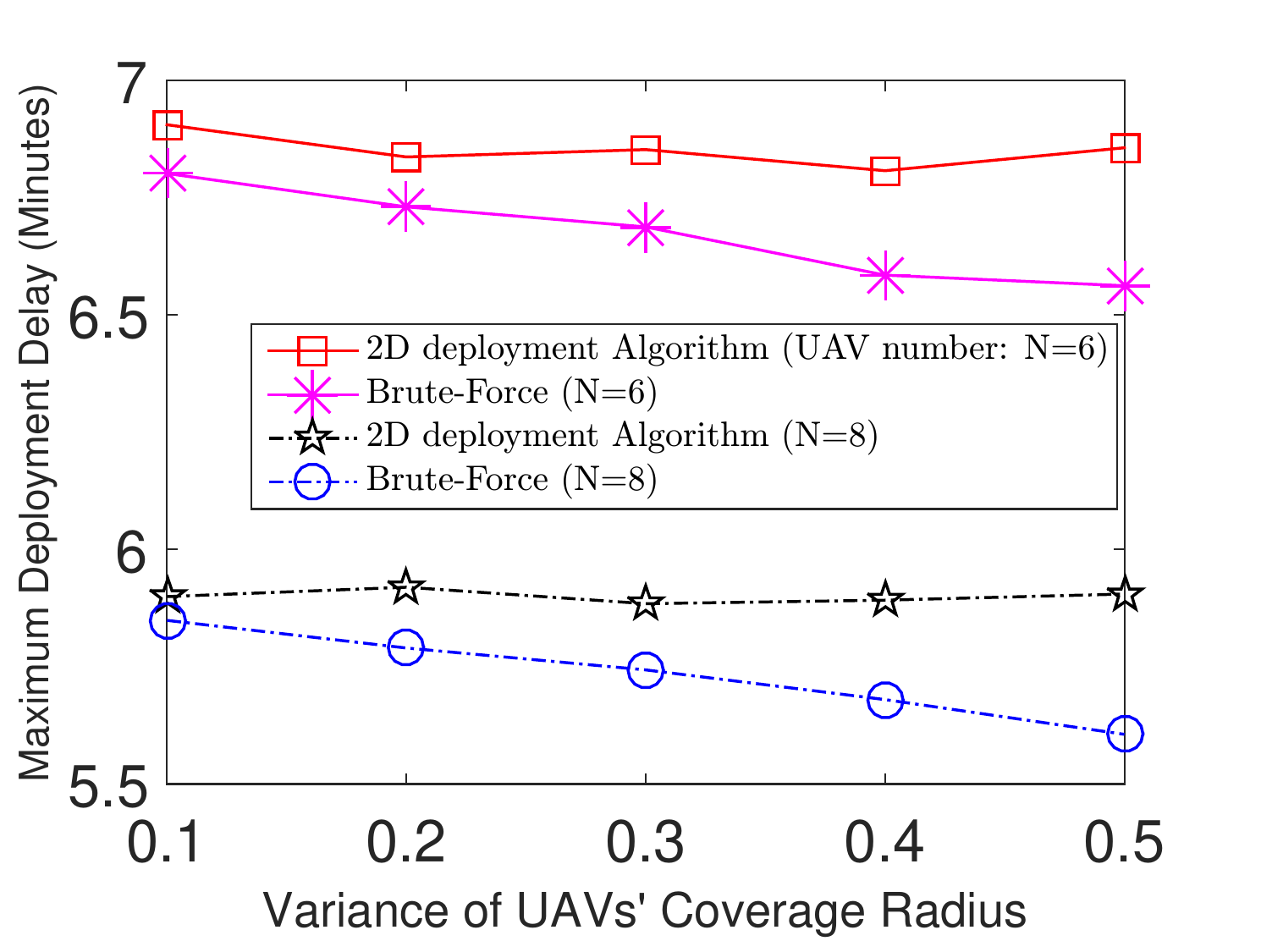}
    \caption{The deployment delays obtained by 2D deployment Algorithm and Brute-Force Algorithm for the min-max problem in 2D area.}
    \label{fig:2Dplot}
\end{figure}

\subsubsection{{Different coverage radius}}
{Now we look at the general case when UAVs have different coverage radii and extend 2D deployment Algorithm to solve it. As the general case is difficult to solve optimally, we view each UAV's coverage radius the same (equal to the minimum radius among all UAVs). Then we grid the rectangular target area according to the minimum coverage radius among UAVs. To show the effectiveness of 2D deployment Algorithm for this case, we compare it with the optimal solution obtained by Brute-Force algorithm. In this experiment, we set $\epsilon = 0.1\%$ as in Proposition \ref{prop:area2D}. The 2D area is set as a square with length of 4 km and width of 4 km. The average flying speed is 20 km/hour, and the minimum and mean coverage radius is set to be 1.5 km and 2 km to guarantee the full coverage. Figure~\ref{fig:2Dplot} shows the Maximum Deployment Delay under \emph{2D deployment Algorithm} versus the variance of UAVs' coverage radius and compares with the optimum obtained by brute-force. We can see that the performance gap between \emph{2D deployment Algorithm} and the brute-force algorithm is small especially when the variance of UAV's coverage radius is small. As \emph{2D deployment Algorithm} views all UAVs the same, this gap enlarges as the variance of UAVs' coverage radius increases. As we have more UAVs or larger $N$, the maximum delay reduces.} 

\section{Conclusion} \label{sec:dis}

The fast deployment of heterogeneous UAVs to provide wireless coverage is of great practical importance. To the best of our knowledge, this is the first work to deal with the emergency criteria of minimization of the maximum deployment delay and the total deployment delay among all UAVs till covering the whole target area. We prove that both min-max and min-sum problems are NP-complete in general. On one hand, when a number $n$ of diverse UAVs are dispatched from the same location, we present an optimal deployment algorithm of computational complexity $O(n^2)$ for the min-max problem. When UAVs are in general dispatched from different locations, by preserving UAVs' location order, we successfully design an FPTAS of computation complexity $O(n^2 \log \frac{1}{\epsilon})$. On the other hand, for the min-sum problem when UAVs are dispatched from the same location, we present an approximation algorithm runs in linear time. As for the general case, we further reformulate it as a dynamic program and propose a pseudo polynomial-time algorithm to solve it optimally. The theoretical results draw in this paper are further confirmed by simulation.

{The interference among UAVs' ground user services will be considered in future work, in which UAVs' coverage radius could be reduced accordingly and we need more UAVs to deploy for full coverage.}



\bibliographystyle{plainnat}

\clearpage

\appendices
\section{Proof of Theorem~\ref{thry:np1}} \label{ap:npc1}
\begin{proof}
We first define the decision version of min-max deployment delay problem in (\ref{general}) as follows: given an integer $K$ as the maximum deployment delay among UAVs, we want to determine whether UAVs can be moved to reach a full coverage within deadline $K$. We call it deployment feasibility problem, which will be proved to be NP-complete.

\begin{figure}[t]
\begin{center}
\includegraphics[scale=0.75]{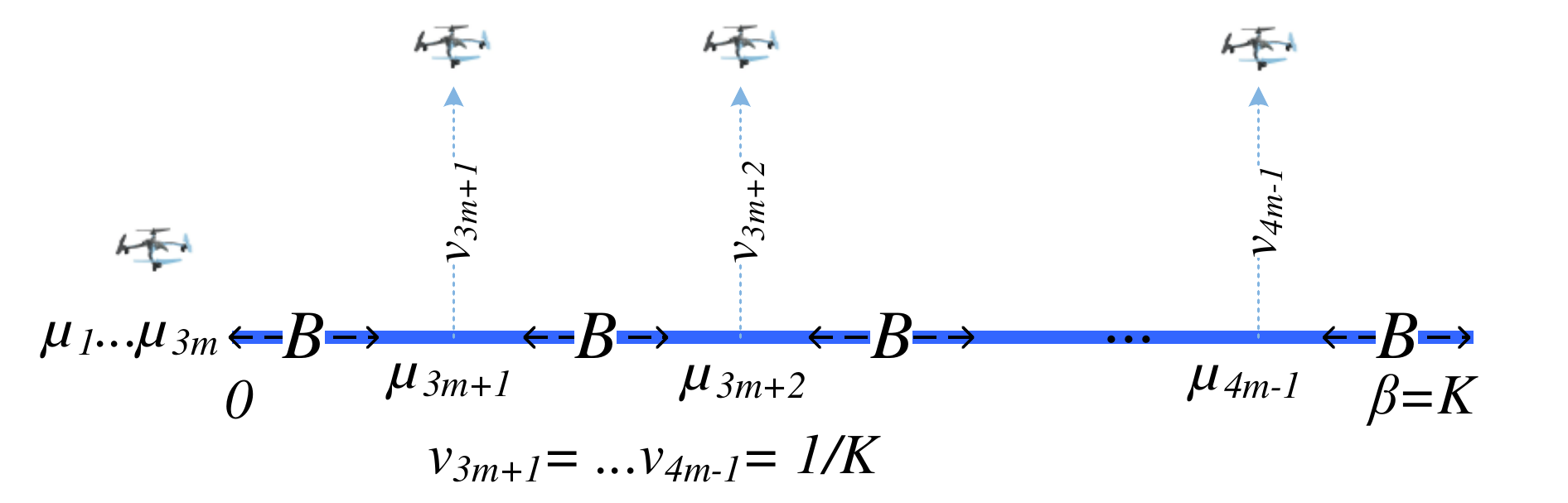}
\caption{Transformation of an arbitrary instance of 3-partition problem to an instance of the deployment feasibility problem.}
\label{fig:npc}
\end{center}
\end{figure}

We next reduce the {\em 3-partition problem} \cite{garey2002computers}, which is a well-known NP-complete problem, to the deployment feasibility problem. In the 3-partition problem, we are given a multiset $M = \{ a_1 \geq a_2 \geq \cdots \geq a_{3m}\}$ of $3m$ positive integers such that $\sum_{i=1}^{3m} a_i = mB$ for some $B$ and $B/4 < a_i < B/2$ for $1 \leq i \leq 3m$. The problem is to decide whether $M$ can be partitioned into $m$ triples $M_1, M_2, \ldots, M_m$ such that the sum of the three numbers in each triple is equal to $B$.

We next transform an arbitrary instance of the 3-partition problem to an instance of the deployment feasibility problem. Let $\beta = mB+m-1$, and $K = \beta$, where $\beta$ is the rightmost point of the target line interval $L = [0, \beta]$ by restricting $d \to 0$ as shown in Figure~\ref{fig:npc}. In the deployment feasibility problem, we have $n = 4m-1$ UAVs in total, which consists of two groups of UAVs. Specifically, for the first group where $1 \leq i \leq 3m$, we construct $3m$ UAVs $\mu_i$ with coverage radius $r_i= a_i/2$ and initial location $(\frac{-a_i}{2}, 0)$. In addition, for the second group, we construct $m-1$ UAVs $\mu_{3m+1}, \mu_{3m+2}, \ldots, \mu_{4m-1}$ of coverage radius $1/2$ and initial locations $(B+1/2, 0), (2B+3/2, 0),(3B+5/2, 0), \ldots, ((m-1)B+(2m-3)/2, 0)$. By construction, we can see that the sum of coverage radii of $4m-1$ UAVs equals to $\frac{1}{2} \beta$. That is to say, all constructed $4m-1$ UAVs should be moved to cover the target interval without overlapping. Otherwise, full coverage can be not achieved. As for UAV $\mu_{i}$, in which $i \in \{3m+1, 3m+2, \ldots, 4m-1\}$, let $h_i = 1$ and flying speed $v_i =\frac{1}{K}$. While for $1 \leq i \leq 3m$, $h_i = 0$ and $v_i = 1$.

Given the transformation above, we next prove that there exists a solution to the instance of the 3-partition problem {\em if and only if} the constructed instance of the deployment feasibility problem is feasible, i.e., UAVs can be moved to reach a full coverage within deadline $K$.
\begin{itemize}
  \item Given a solution to the instance of the 3-partition problem, i.e., there is a partition of $m$ triples $M_1, M_2, \ldots, M_m$, where the sum of each triple being $B$. We can move $m-1$ UAVs in the second group from initial locations $(B+1/2, 0), (2B+3/2, 0),(3B+5/2, 0), \ldots, ((m-1)B+(2m-3)/2, 0)$ to $(B+1/2, 1), (2B+3/2, 1),(3B+5/2, 1), \ldots, ((m-1)B+(2m-3)/2, 1)$ vertically. It can be observed that the deployment delay is $\frac{h_i}{v_i} = K$ since the vertical distance and flying speed is $1$ and $\frac{1}{K}$ respectively. Then, move three UAVs in the first group of each triple $M_i$ to cover each block with exactly length $B$ as shown in Figure~\ref{fig:npc}. Since each UAV $\mu_i$ in the first group are initially located at $(\frac{-a_i}{2}, 0)$ with coverage radius $r_i= a_i/2$, the longest horizontal distance it needs to move along $L$ is $\beta$. Since $\beta = K$, and for $1 \leq i \leq 3m$, $h_i = 0$ and $v_i = 1$, the deployment delay of the three UAVs corresponding to $M_i$ moving horizontally to cover the block of length $B$ is at most $K$. Therefore, the target interval is fully covered by UAVs with the maximum deployment delay $K$.
  \item Now, we have a feasible solution to instance of the deployment feasibility problem, where UAVs are moved to reach a full coverage within the deadline $K$ (i.e., it is feasible). Within the deadline $K$, we first observe that we can only move the UAVs in the second group vertically to hovering height $1$. Since if it moves with any horizontal distance $x$, the deployment delay will be $\sqrt{x^2 + 1} \cdot K$, which is larger than $K$. Thus, $[B,B+1] \cup [2B+1, 2B+2] \cup \ldots \cup [(m-1)B+m-1, mB+m-1]$ are covered by only moving each UAV in the second group vertically to hovering height $1$. Hence, there are $m$ uncovered blocks with exactly length $B$, i.e., $[0,B] \cup [B+2, 2B+2] \cup \ldots \cup [(m-1)B+2m-2, mB+2m-2]$. Thus, for the remaining $m$ uncovered blocks of length $B$, we can only move the UAVs in the first group to fill by partitioning $M$ into $m$ triples $M_1, M_2, \ldots, M_m$. We know that the sum of each triple is $B$ and $\sum_{i=1}^{3m} a_i = mB$, which implies that we have a solution to the instance of the 3-partition problem.
\end{itemize}

{By adding one additional UAV with radius $1$ and flying speed $\frac{1}{K+1}$ initially located at $(-\frac{3}{2}, 0)$, we can create an instance of the problem where $\sum_{i=1}^{4m} 2 r_i > \beta$, and the proof still holds since the additional UAV travels $K+1$ time units for distance $1$ and it can not be dispatched to cover the interval.}

We have proved that if there exists a solution to the instance of the 3-partition problem, the constructed instance of deployment feasibility problem has a solution within deadline $K$. Conversely, if the constructed instance of deployment feasibility problem has a feasible solution within deadline $K$, there exists a corresponding feasible solution to the instance of the 3-partition problem. The deployment feasibility problem is in NP, since we can verify whether any UAV deployment fully covers the target interval or not in polynomial time, which completes our proof.
\end{proof}

\section{Proof of Proposition~\ref{common}} \label{ap:common}
\begin{proof}
We first note that ${T}$ produced by Algorithm~\ref{alg:common} is feasible. Otherwise, Algorithm~\ref{alg:common} ends when $\overline{\beta} > 0$, which contradicts with our assumption $2 \sum_{i=1}^{n}{r_{i}} \geq \beta$. Thus, ${T}$ is a feasible solution. Since we are looking for minimum maximum travel time among all available UAVs to cover $[0, \beta]$, the larger the distance from the initial location to the target position, the more travel time needs. Thus, given all available UAVs, we first consider using one UAV to cover the furthest point of the target area, but not exceeding it. Otherwise, it can always move closer to reduce the travel time.

\begin{figure}[t]
    \centering
        \includegraphics[width=0.75\textwidth]{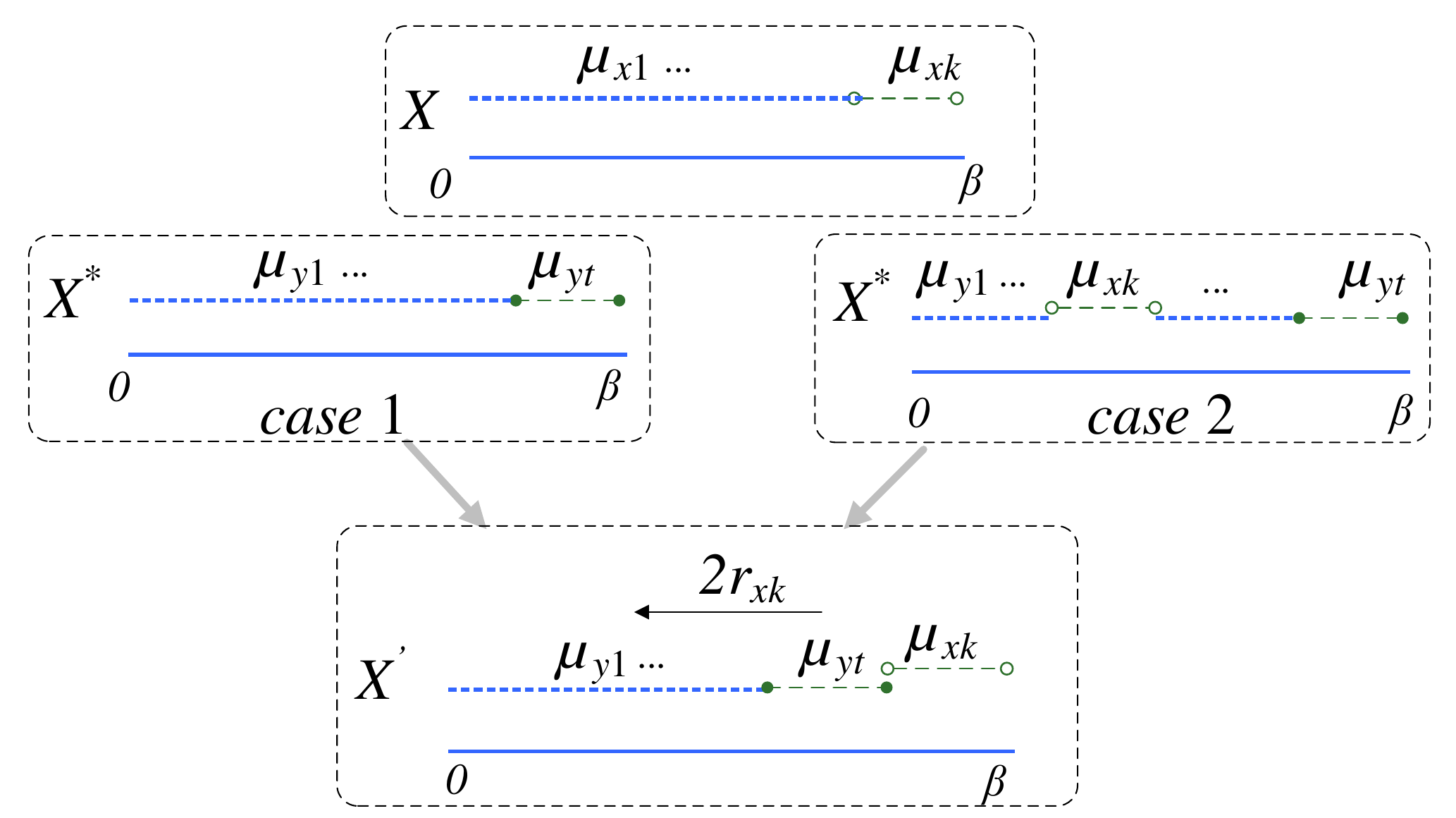}
    \caption{Optimality proof illustration of Algorithm~\ref{alg:common}.}
    \label{fig:greedyproof}
\end{figure}

As shown in Figure~\ref{fig:greedyproof}, suppose we have our solution $X$ with UAV sequence\footnote{Without loss of generality, assume that the UAVs are ordered from the closest point to the furthest point (left to right in our case).} $\{\mu_{x1}, \mu_{x2}, \cdots, \mu_{xk}\}$ and optimal solution $X^*$ with UAV sequence $\{ \mu_{y1}, \mu_{y2}, \cdots, \mu_{yt}\}$, where $k \leq n$ and $t \leq n$. They have the maximum delay of $T$ and $T^*$, respectively. Next, we will show that if $X \neq X^*$, we can always adjust the sequence of $X^*$ to be identical to $X$ without increasing the delay.

We first consider the furthest point of the target area, it is covered by $\mu_{yt}$ in $X^{*}$ while covered by $\mu_{xk}$ in $X$. If $\mu_{yt}$ and $\mu_{xk}$ are the same one, then nothing to prove. Otherwise, $T_{yt} > T_{xk}$ since Algorithm~\ref{alg:common} in Line 5 computes and selects the UAV $\mu_{xk}$ with the minimum travel time $T_{xk}$ among all available UAVs. Next, we try to adjust the order of $\mu_{yt}$ and $\mu_{xk}$ (if exists) in $X^*$ according to Algorithm~\ref{alg:common}. There will be two cases that $\mu_{xk}$ is in $\{ \mu_{y1}, \mu_{y2}, \cdots, \mu_{yt}\}$ or not, see Figure  ~\ref{fig:greedyproof}.

We first consider that $\mu_{xk}$ is not in $\{ \mu_{y1}, \mu_{y2}, \cdots, \mu_{yt}\}$ (case 1 in Figure~\ref{fig:greedyproof}), then we can just add $\mu_{xk}$ to cover the furthest point of the target area and push all the other UAVs to left $2r_{xk}$. Thus, we have the new UAV sequence $\{\mu_{y1}, \mu_{y2}, \cdots, \mu_{yt}, \mu_{xk}\}$. In this new sequence $X^{\prime}$, the maximum delay $T^*$ will not increase, since $T_{yt} > T_{xk}$ in $X^*$ and $\{\mu_{y1}, \mu_{y2}, \cdots, \mu_{yt}\}$ are moved to left by $2r_{xk}$. Next, for the other case that $\mu_{xk}$ is in $\{ \mu_{y1}, \cdots, \mu_{xk}, \cdots, \mu_{yt}\}$ (case 2 in Figure~\ref{fig:greedyproof}), we can still move $\mu_{xk}$ to the right of $\mu_{yt}$ and push all the other UAVs to left by $2r_{xk}$ without incurring more delay. Specifically, in $\{ \mu_{y1}, \cdots, \mu_{xk}, \cdots, \mu_{yt}\}$, if we move $\mu_{xk}$ to the right of $\mu_{yt}$ and push the other UAVs to left, the UAVs to the left of $\mu_{xk}$ will stand still (without incurring more delay) and the UAVs to the right of $\mu_{xk}$ will move to left by $2r_{xk}$ (incurring less delay). Therefore, in the first iteration, i.e., covering the furthest point of the target, we can always use the UAV with the minimum travel time among all available UAVs as given by Algorithm~\ref{alg:common}. Similarly, our algorithm produces the minimum travel time for each iteration with the available UAVs. Therefore, $T \leq T^*$, which shows that Algorithm~\ref{alg:common} is optimal.

During each UAV deployment (line 5 of Algorithm~\ref{alg:common}), we face at most $n$ UAVs in set $|\textbf{U}^-|$ and choose the best in each while loop, which starts from $n$ and decrease by one in each iteration. Thus, Algorithm~\ref{alg:common} runs in $O(n^2)$ time, which completes our proof.
\end{proof}

\section{Proof of Proposition~\ref{lm:decalg}} \label{ap:dec}
\begin{proof}
We first show that if the algorithm outputs that $T$ is feasible, then the computed solution is feasible. First, we notice that $y_{j} < y_{k}$ for each $\mu_{i} \in \textbf{U}^-$ if and only if $x_{j} < x_{k}$. Consider the $i$-th iteration. If $\overline{\beta} \notin [a_{i}, b_{i}]$, then $\mu_i$ is not selected. Otherwise, if $\overline{\beta}+r_i < b_i -r_i$, then $y_i = \overline{\beta} + r_i$. There is no such UAV $\mu_j$ such that $y_j > y_i$ and $j < i$, since $\overline{\beta} > y_j$ for all the UAV $\mu_j$ with $j<i$. In the other case when $\overline{\beta}+r_i > b_i -r_i$, then $y_i = b_i-r_i$. Suppose there exists some $j < i$ such that $y_j > y_i$, we know that $y_i + r_i \geq y_j+ r_j \Longrightarrow 0 < y_j - y_i \leq r_i -r_j$. Thus, we have $r_i -r_j > 0 > y_i -y_j \Longrightarrow y_j - r_j > y_i - r_i$. Thus, we can remove $\mu_j$ from $\textbf{U}^{-}$ without any loss to the current covered interval.

Then, let $\overline{\beta}_{i}$ denote the value of $\overline{\beta}$ after the $i$-th iteration. Initially, $\overline{\beta}_{0} = 0$. We prove by induction on $i$ that $[0, \overline{\beta}_{i}]$ is covered. Consider iteration $i$. If $\overline{\beta}_{i-1} \in [y_{i}-r_i, y_{i}+r_i]$, then $\overline{\beta}_{i} = y_{i} + r_i$.Next, we show that if the algorithm outputs that $T$ is not feasible, there is no feasible solution. We prove by induction on $[0, \overline{\beta}_{i}]$ is the longest interval that can be covered by UAVs $\mu_{1},\ldots,\mu_{i}$. In the base case, observe that $\overline{\beta}_{0} = 0$ is optimal. For the induction step, let ${C^{\prime}_{i}}$ be a sequence of UAVs $\mu_{1},\ldots,\mu_{i}$ that covers the interval $[0, \overline{\beta}^{\prime}_{i}]$. Let $[0, \overline{\beta}^{\prime}_{i-1}]$ be the interval that ${C^{\prime}_{i-1}}$ covers by $\mu_{1},\ldots,\mu_{i-1}$. By the inductive hypothesis, $\overline{\beta}^{\prime}_{i-1} \leq \overline{\beta}_{i-1}$. If $\overline{\beta}^{\prime}_{i-1} \leq \overline{\beta}_{i-1} < a_{i}$ or $v_i T < h_i$, it follows that $\overline{\beta}^{\prime}_{i} = \overline{\beta}^{\prime}_{i-1} \leq \overline{\beta}_{i-1} = \overline{\beta}_{i}$. Otherwise, we have $y_{i} = \min\{\overline{\beta}_{i-1} + r_{i}, b_{i} - r_i\}$ and $y^{\prime}_{i} = \min\{\overline{\beta}^{\prime}_{i-1} + r_{i}, b_{i} - r_i\}$. Observe that $y^{\prime}_{i} \leq y_{i}$ and therefore $\overline{\beta}^{\prime}_{i} \leq \overline{\beta}_{i}$.

Finally, with respect to the running time, $a_{i}$ and $b_{i}$ for $1 \leq i \leq n$ can be computed in $O(n)$ time. There are $n$ iterations for $n$ UAVs, each taking at most $O(n)$ time in line 10. Hence, the time complexity of Algorithm~\ref{alg:decalg} is $O(n^2)$.
\end{proof}

\section{Proof of Theorem~\ref{thry:np2}} \label{ap:npc2}
\begin{proof}
We define the decision version of total deployment delay minimization problem in (\ref{sum}) as follows: given an integer $K$, the problem is to determine whether UAVs can be moved to reach a full coverage such that the sum of all UAVs' deployment delay is at most $K$. We call it total delay feasibility problem. Then, we prove that total delay feasibility problem is NP-complete. We only describe the key points, since the reduction is similar to Theorem~\ref{thry:np1}.

\begin{figure}[t]
\includegraphics[scale=0.75]{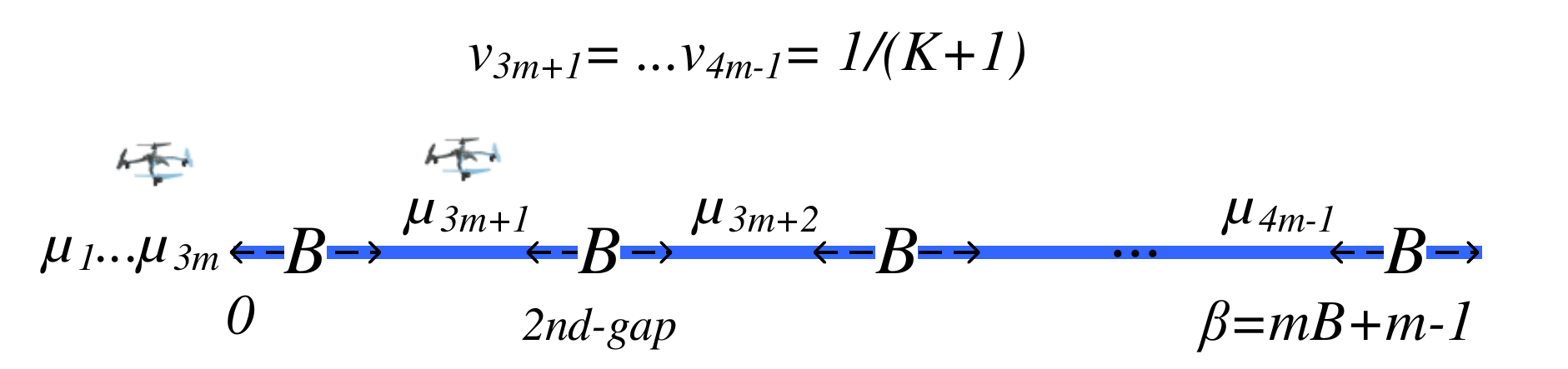}
\caption{Transformation of 3-partition problem to the total deployment delay problem.}
\label{fig:npc2}
\end{figure}

As shown in Figure~\ref{fig:npc2}, we transform an arbitrary instance of 3-partition problem to an instance of the total delay feasibility problem. The line interval is denoted as $L = [0, \beta]$, and let $\beta = mB+m-1$ and $K = 3Bm(m+1)+3(m-1)$. In the total delay feasibility problem, we have $n = 4m-1$ UAVs in total, which consists of two groups of UAVs. Specifically, for the first group where $1 \leq i \leq 3m$, we construct $3m$ UAVs $\mu_i$ with coverage radius $r_i = a_i/2$ and initial location $(\frac{-a_i}{2}, 0)$. In addition, for the second group, we construct $m-1$ UAVs $\mu_{3m+1}, \mu_{3m+2}, \ldots, \mu_{4m-1}$ of coverage radius $1/2$ and initial locations $(B+1/2, 0), (2B+3/2, 0),(3B+5/2, 0), \ldots, ((m-1)B+(2m-3)/2, 0)$. By construction, we can see that the sum of coverage radii of $4m-1$ UAVs equals to $\frac{1}{2} \beta$. That is to say, all constructed $4m-1$ UAVs should be moved to cover the target interval without overlapping. Otherwise, full coverage can be not achieved. As for UAV $\mu_{i}$, in which $i = \{3m+1, 3m+2, \ldots, 4m-1\}$, let $h_i = 0$ and $v_i =\frac{1}{K+1}$. While for $1 \leq i \leq 3m$, $h_i = 0$ and $v_i = 1$.

We now prove that there exists a solution to the instance of the 3-partition problem {\em if and only if} the constructed instance of total delay feasibility problem has a solution of at most $K$ deployment delay.

Given a solution to the instance of the 3-partition problem, i.e., there is a partition of $m$ triples $M_1, M_2, \ldots, M_m$, the sum of each triple being $B$. Because any UAV in second group moving at least distance $1$ will incur delay of $K+1$. We just keep the $m-1$ UAVs in the second group in the initial positions, which results in total deployment delay $0$. Then, without loss of generality, we assume that we move the triple $M_i$ with three UAVs to fill $i$-th gap of length $B$ from left to right as shown in Figure~\ref{fig:npc2}. For each $M_i$, the total travel time of the three UAVs moving into $i$-th gap is less than $3iB + 3i-3$. Thus, the sum of deployment delay of all UAVs is thus less than $\frac{3Bm(m+1)+3(m-1)}{2} = K/2$ in this case. Therefore, we have a solution of at most $K$ to the instance of the total delay feasibility problem. Conversely, it is also true that there is a feasible solution of at most $K$ to the total delay feasibility problem implies that we have a solution to the instance of the 3-partition problem.

{Similar to the proof of Theorem~\ref{thry:np1}, we can add one additional UAV with radius $1$ and flying speed $\frac{1}{K+1}$ initially located at $(-\frac{3}{2}, 0)$ to create an instance of the problem where $\sum_{i=1}^{4m} 2 r_i > \beta$.}
\end{proof}

\section{Proof of Lemma~\ref{fast}} \label{ap:fast}
\begin{proof}
{It can be proved by Greedy Exchange. Suppose that there is a feasible solution containing UAVs $\mu_a$ and $\mu_b$ when all UAVs are dispatched from $0$. As we know, all UAVs are with identical flying speed and operating altitude. As shown in Figure~\ref{fig:exchange}, there are two UAVs $\mu_a$ and $\mu_b$ with radii $r_a$ and $r_b$, in which $r_a > r_b$, and $\mu_a$ is located to the left of $\mu_b$. We can always swap the positions of $\mu_a$ and $\mu_b$ incurring less total delay. Moreover, the middle part (blue line) also shifts to left incurring less total delay. Overall, we can always deploy a UAV with longer wireless coverage radius to further location for saving the travel distance and delay.}
\end{proof}

\begin{figure}[t]
    \centering
        \includegraphics[width=0.75\textwidth]{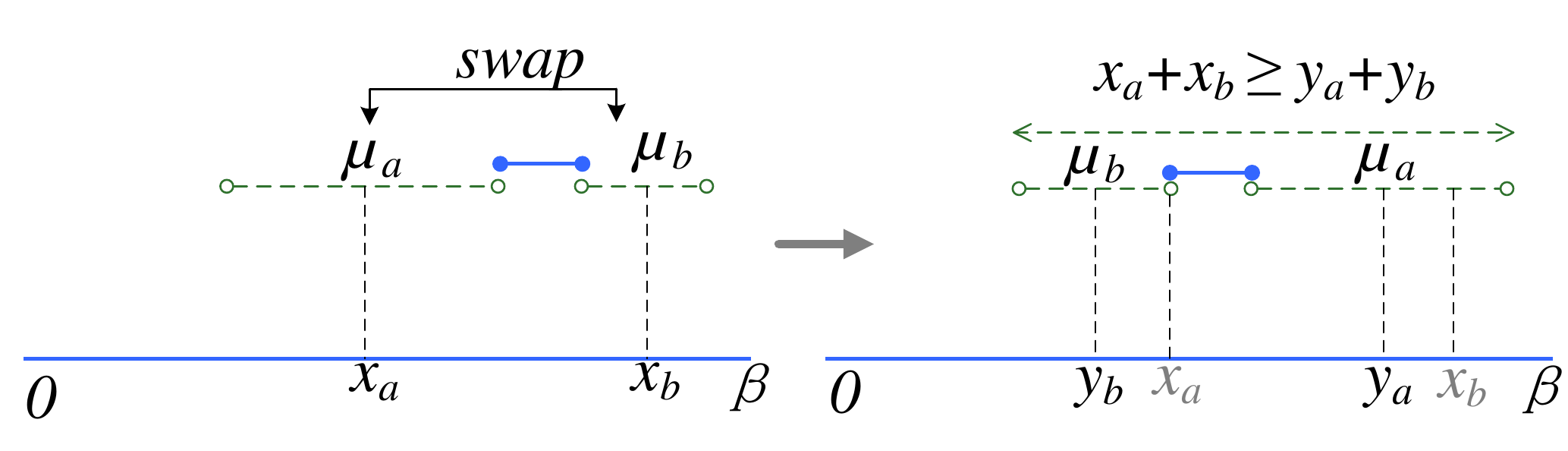}
    \caption{Optimality proof illustration of Lemma \ref{fast}}.
    \label{fig:exchange}
\end{figure}

\end{document}